\documentclass{aims}
\usepackage{amsmath}
  \usepackage{paralist}
  \usepackage{graphics} %% add this and next lines if pictures should be in esp format
  \usepackage{epsfig} %For pictures: screened artwork should be set up with an 85 or 100 line screen
\usepackage{graphicx}  \usepackage{epstopdf}%This is to transfer .eps figure to .pdf figure; please compile your paper using PDFLeTex or PDFTeXify.
 \usepackage[colorlinks=true]{hyperref}
\hypersetup{urlcolor=blue, citecolor=red}

\def\currentvolume{X}
 \def\currentissue{X}
  \def\currentyear{200X}
   \def\currentmonth{XX}
    \def\ppages{X--XX}

\usepackage{mathtools}
\usepackage{tikz-network}
\usepackage{stmaryrd}
\usepackage{booktabs}

\DeclarePairedDelimiter{\floor}{\lfloor}{\rfloor}

\newcommand*{\QEDB}{\hfill\ensuremath{\square}}

\newcommand{\ie}{\emph{i.e.}}

\newcommand\dsb[1]{\llbracket #1 \rrbracket}
\newtheorem{theorem}{Theorem}[section]

\newtheorem{lemma}[theorem]{Lemma}
\newtheorem{proposition}{Proposition}

\theoremstyle{definition}
\newtheorem{definition}[theorem]{Definition}

\newtheorem{example}{Example}

\title[Improved Quantum Codes from Metacirculant Graphs] %Use the shortened version of the full title
      {Improved Quantum Codes from Metacirculant Graphs via Self-Dual Additive $\mathbb{F}_4$-Codes}

\author[Padmapani Seneviratne and Martianus Frederic Ezerman]{}


 \email{Padmapani.Seneviratne@tamuc.edu}
 \email{fredezerman@ntu.edu.sg}
\thanks{Nanyang Technological University Grant Number 04INS000047C230GRT01 supports the research carried out by M.~F.~Ezerman.}

\thanks{$^*$ Corresponding author: Martianus Frederic Ezerman}

\begin{document}
\maketitle

% Enter the first author's name and address:
\centerline{\scshape Padmapani Seneviratne}
\medskip
{\footnotesize
% please put the address of the first author
 \centerline{Department of Mathematics, Texas A\&M University-Commerce}
%   \centerline{Other lines}
   \centerline{2600 South Neal Street, Commerce TX 75428, USA}
} % Do not forget to end the {\footnotesize by the sign }

\medskip

\centerline{\scshape Martianus Frederic Ezerman$^*$}
\medskip
{\footnotesize
 % please put the address of the second  and third author
 \centerline{School of Physical and Mathematical Sciences, Nanyang Technological University}
%   \centerline{Other lines}
   \centerline{21 Nanyang Link, Singapore 637371}
}

\bigskip

% The name of the associate editor will be entered by an editorial staff
% "Communicated by the associate editor name" is not needed for special issue.
 \centerline{(Communicated by the associate editor name)}

%The abstract of your paper
\begin{abstract}
We use symplectic self-dual additive codes over $\mathbb{F}_4$ obtained from metacirculant graphs to construct, for the first time, $\dsb{\ell, 0, d}$ qubit codes with parameters $(\ell,d) \in \{(78, 20), (90, 21), (91, 22), (93,21),(96,21)\}$. Secondary constructions applied to the qubit codes result in many qubit codes that perform better than the previous best-known.
\end{abstract}

%The title of your section 1
\section{Introduction}\label{sec:intro}

We work on three closely connected objects, namely, a \emph{metacirculant graph} $G$, a \emph{symplectic self-dual additive code} $C$ over $\mathbb{F}_4$, and its corresponding \emph{quantum stabilizer code} $Q$. The route is straightforward. Let $I$ be the identity matrix of a suitable dimension and $\omega$ be a root of $x^2+x+1 \in \mathbb{F}_2[x]$. We search for a $G$ whose adjacency matrix $A(G)$ leads to $C$, which is generated by the row span of $A(G)+ \omega I$. This code $C$, in turn, yields $Q$ via the stabilizer method. 

For lengths $\ell \in \{27,36\}$ exhaustive searches are feasible, allowing us to construct some families of additive codes that contain some code $C$ with an improved minimum distance when compared with the best that circulant graphs can lead to. For lengths $\ell \in \{78, 90, 91, 93, 96\}$ we run randomized non-exhausted searches to come up with the improved additive codes. We exhibit numerous instances when the resulting qubit codes $Q$ have better parameters than the best-known comparable quantum codes. 

The general construction of quantum stabilizer codes, wherein classical codes are used to describe the quantum error operators, is well-established. Our main reference for the qubit case is the seminal work of Calderbank, Rains, Shor, and Sloane in \cite{Calderbank1998}. Two recent introductory expositions can be found in \cite{Grassl20} and \cite{Ezerman2021}.

To get to metacirculant graphs, we recall \emph{circulant graphs}, which have been more extensively studied. A recent survey on the subject can be found in \cite{Monakhova2012}. The following definition of circulant graphs is given in \cite{Bogdanowicz1996}. 

\begin{definition}
Let $\mathbb{Z}_{n}$ denote the ring of integers modulo $n$ and let
\[
\mathbb{Z}_{n}^{*} := \{x \in \mathbb{Z}_n : 0 < x \le n/2 \}.
\]
 The circulant graph $\Gamma_{n}(S)$ is the graph with the vertex set $\mathbb{Z}_{n}$ where any two vertices $x$ and $y$ are adjacent if and only if $|x-y|_{n} \in S$, with $S \subseteq \mathbb{Z}_{n}^{*}$ and 
\[
\begin{cases}
|a|_{n} := |a| & \mbox{if } 0 \le a \le n/2, \\
|a|_{n} := n-a & \mbox{if } n/2 < a < n.
\end{cases}
\]
\end{definition}

The adjacency matrix of a circulant graph is a circulant matrix. An $n \times n$ matrix $A$ is called a {\it circulant matrix} if it has the form
\begin{equation}~\label{eq:one}
	A = 
	\begin{pmatrix}
		a_1 & a_2 & \cdots & a_{n-1} & a_n \\
		a_n & a_1 & \cdots & a_{n-2} & a_{n-1} \\
		\vdots  & \vdots  & \ddots & \vdots & \vdots \\
		a_3 & a_4 & \cdots & a_1 &a_2 \\
		a_2 & a_3 & \cdots & a_n &a_1 
	\end{pmatrix}.
\end{equation}	

If the adjacency matrix $A:=A(G)$ of $G$ is circulant, then 
$a_1 = 0$ and $a_i = a_{n+2-i}$ for $i \in \{2,\ldots, \floor{ n/2}\}$. Circulant matrices are used as building blocks in the constructions of many different classes of codes. Examples include self-dual codes, cyclic codes, and quadratic residue codes.

Circulant graphs are {\it vertex transitive} \cite{Bogdanowicz1996}. They are the Cayley graphs of $\mathbb{Z}_n$. In 1982, Alspach and Parsons~\cite{Alspach1982} constructed a family of vertex transitive graphs. Each graph in the family has a transitive permutation group as a subgroup of its automorphism group. They named the family {\it metacirculant graphs} as it contains the class of circulant graphs. 

Li {\it et al.} in \cite[Definition 1.1]{LSW2013}, following D. Maru\v{s}i\v{c} in \cite{Marusic2003}, call a graph $\Gamma=(V,E)$ an $(m,n)$-metacirculant if $|V|= m n$ and $\Gamma$ has two automorphisms $\rho$ and $\sigma$ that satisfy some conditions. First, $\left\langle \rho \right\rangle$ is semiregular and has $m$ orbits on $V$. Second, $\sigma$ cyclically permutes the $m$ orbits of $\left\langle \rho \right\rangle$ and normalizes $\left\langle \rho \right\rangle$. Third, $\sigma^m$ fixes at least one vertex of $\Gamma$. In this work, we follow an equivalent combinatorial definition.

\begin{definition} (\cite{Alspach1982})\label{def:mc}
Let $m,n$ be two fixed positive integers and $\alpha\in \mathbb{Z}_{n}$ be a unit. Let $S_{0}, S_{1}, \ldots, S_{\floor{m/2}} \subseteq \mathbb{Z}_{n}$ satisfy the four properties
\begin{enumerate}
	\item $S_{0} = -S_{0}$.
	\item $0\notin S_{0}$.
	\item $\alpha^{m}S_k = S_k$ for $1\le k\le \lfloor m/2 \rfloor$.
	\item If $m$ is even then $\alpha^{m/2}S_{m/2} = -S_{m/2}$. 
\end{enumerate}
The meta-circulant graph $\Gamma:=\Gamma\left(m, n, \alpha, S_{0}, S_{1}, \ldots, S_{\floor{m/2}}\right)$ has the vertex set $V(\Gamma)=\mathbb{Z}_{m}\times \mathbb{Z}_{n}$. Let $V_0, V_1, \ldots, V_{m-1}$, where 
$V_{i} := \{(i, j) : 0\le j \le n-1 \}$ is a partition of $V(\Gamma)$. Let $1 \le k \le \floor{m/2}$. Vertices $(i, j)$ and $(i+k, h)$ are adjacent if and only if $(h-j) \in \alpha^{i} \, S_k$.
\end{definition}

Henceforth, we assume $m > 1$, since a metacirculant graph with $m=1$ is a circulant graph.

\begin{example}
The Petersen graph is $\Gamma(2,5,2,\{1,4\}, \{0\})$. The vertices are partitioned into 
\[
V_0 :=\{(0,0),(0,1),(0,2),(0,3),(0,4)\} \mbox{ and }
V_1 :=\{(1,0),(1,1),(1,2),(1,3),(1,4)\},
\]
with adjacent vertices $(0,j)$ and $(1,j)$, for $0 \leq j \leq 4$. Figure~\ref{fig:peter} relabels the vertices lexicographically, that is, the upper layer contains the vertices in $V_0$ as $1, \ldots, 5$ while the lower layer presents the vertices in $V_1$ as $6, \ldots, 10$.
	
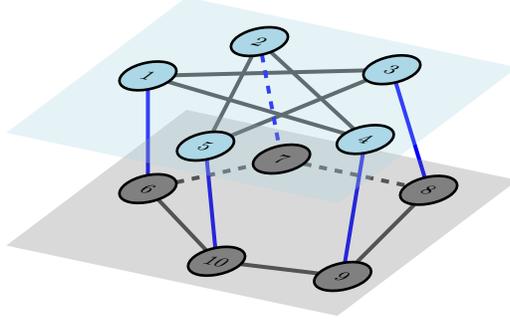
\begin{figure}[h]\label{fig:peter}
		\centering
		\begin{tikzpicture}[multilayer=3d]
			\SetLayerDistance{-1.5}
			\Plane[x=-.7,y=-1.5,width=4.5,height=3,color=gray,layer=2,NoBorder]
			\Plane[x=-.7,y=-1.5,width=4.5,height=3,NoBorder]
			\Vertex[IdAsLabel,layer=1]{1}
			\Vertex[x=0.7,y=1,IdAsLabel,layer=1]{2}
			\Vertex[x=2.5,y=1,IdAsLabel,layer=1]{3}
			\Vertex[x=3.2,y=-0.3,IdAsLabel,layer=1]{4}
			\Vertex[x=1.6,y=-1,IdAsLabel,layer=1]{5}
			% second layer	
			\Vertex[IdAsLabel,color=gray,layer=2]{6}
			\Vertex[x=1,y=1,IdAsLabel,color=gray,layer=2]{7}
			\Vertex[x=3.0,y=1,IdAsLabel,color=gray,layer=2]{8}
			\Vertex[x=3.3,y=-0.8,IdAsLabel,color=gray,layer=2]{9}
			\Vertex[x=1.75,y=-1,IdAsLabel,color=gray,layer=2]{10}
			%% edge list	
			\Edge[style=dashed](6)(7)
			\Edge(6)(10)
			\Edge(1)(4)
			\Edge(1)(3)
			\Edge(2)(5)
			\Edge(2)(4)
			\Edge(3)(5)
			\Edge[style=dashed](7)(8)
			\Edge(8)(9)
			\Edge(9)(10)
			\Edge[color=blue](1)(6)
			\Edge[style=dashed,color=blue](2)(7)
			\Edge[color=blue](5)(10)
			\Edge[color=blue](4)(9)
			\Edge[color=blue](3)(8)
\end{tikzpicture}
\caption{The Petersen Graph as $\Gamma(2,5,2,\{1,4\}, \{0\})$.}
\end{figure}
%\QEDB
\end{example}

The following fact is well-known.

\begin{theorem} {\rm (\cite{Alspach1982})} Let $\phi(\cdot)$ denote the Euler totient function. If $\gcd(m,n) =1$ and $\gcd(m,\phi(n))=1$, then the  metacirculant $\Gamma\left(m, n, \alpha, S_{0}, S_{1}, \ldots, S_{\floor{m/2}}\right)$ is isomorphic to a circulant graph.
\end{theorem}

The next lemma determines the valence.

\begin{lemma}\label{lem:valence}
Any $G:=\Gamma \left(m > 1, n, \alpha, S_{0}, S_{1}, \ldots, S_{\floor{m/2}} \right)$ is regular with valence
\begin{equation}\label{eq:valence}
\begin{cases} 
	|S_0| + |S_1| & \mbox{if } m=2,\\
	|S_0| + |S_{\floor{ m/2}}| + 
	2 \left(\sum_{r=1}^{\floor{m/2}-1} |S_r|\right) & \mbox{if } m \geq 4 \mbox{ is even},\\
	|S_0| + 2 \left(\sum_{r=1}^{\floor{m/2}} |S_r|\right) & \mbox{if } m \geq 3 \mbox{ is odd}.
\end{cases}
\end{equation}
\end{lemma}

\begin{proof}
Since the automorphism group of $G$ acts transitively on its vertices, $G$ is a regular graph. Hence, it suffices to determine the degree of vertex $(0, 0) \in V_{0}$. The vertex is adjacent to a vertex $(0,h) \in V_0$ if and only if $h\in S_0$. This implies that $(0, 0)$ is adjacent to $|S_0|$ vertices in $V_0$. 
	
Next, $(0,0)$ is adjacent to a vertex $(k,h) \in V_k$ if and only if $h \in S_k$. Hence, $(0,0)$ is adjacent to $|S_k|$ vertices in $V_k$. The set $S_k$ determines the edges between layers $V_0$ and $V_{k}$ whose indices differ by $k$ for $1 \leq k \leq \floor{m/2}$. When $m$ is odd, the index $0$ of $V_0$ differs by $k$ from both the indices $k$ and $m-k$ of $V_k$ and $V_{m-k}$ for all $1 \leq k \leq \floor{m/2}$. When $m$ is even, $V_{m/2} = -V_{m/2}$, which implies that the index $0$ of $V_0$ differs by $k$ from both the indices $k$ and $m-k$ of $V_k$ and $V_{m-k}$ for all $1 \leq k \leq \floor{m/2}-1$. Thus, the degree of vertex $(0,0)$ is as given in Equation~\ref{eq:valence}.
\end{proof}

\begin{theorem}\label{thm:mult}
Let $S := \big\{S_0, S_1, \ldots, S_{\floor{m/2}}\big\}$. Let the graph $G:=\Gamma(m,n,\alpha, S)$ be given. Let $\{V_i : i \in \{0, 1, \ldots, m-1\}\}$, with $V_{i} := \{(i, j) : 0 \le j < n \}$,
be the partition of the vertex set into $m$ layers, each containing $n$ vertices. Then $G$ is a multi-partite metacirculant graph with $m$ partitions if and only if $S_0$ is the empty set. 
\end{theorem}

\begin{proof}
 By Definition~\ref{def:mc}, two vertices $(i, j)$ and $(i+k, h)$ are adjacent in $\Gamma(m,n,\alpha,S)$ if and only if $h-j \in \alpha^{i} \, S_k$. Both vertices are in the same layer $V_i$ whenever $k=0$. Hence, the two vertices $(i, j)$ and $(i, h)$ in $V_i$ are adjacent if and only if $h-j \in \alpha^{i} \, S_0$. The set $S_0$ being empty implies $h-j \notin \alpha^{i} \, S_0$. Thus, there is no edge between any pair of vertices within the same layer $V_i$.
\end{proof}

\section{Self-dual additive codes from metacirculant graphs}\label{sec:SDC}

A code over $\mathbb{F}_4 :=\{0,1,\omega,\overline{\omega}= \omega^2=1+\omega\}$ is said to be {\it additive} if it is $\mathbb{F}_2$-linear, \ie, the code is closed under addition but closure under multiplication by the elements in $\mathbb{F}_4 \setminus \mathbb{F}_2$ is not required. An $\mathbb{F}_4$-linear code is additive. An element $\mathbf{c}$ of $C$ is called a codeword of $C$. The {\it weight} of $\mathbf{c}$ is the number of nonzero entries that it has. The {\it minimum distance} of $C$ is the least non-zero weight of all codewords in $C$. If $C$ is an additive code of length $\ell$ over $\mathbb{F}_4$, of size $2^k$ and minimum distance $d$, then we denote $C$ by $(\ell, 2^k, d)_4$.

The {\it trace Hermitian inner product} of $\mathbf{x}=(x_1,\ldots,x_n)$ and $\mathbf{y}=(y_1,\ldots,y_n)$ in $\mathbb{F}_4^n$ is given by
\begin{equation}
	\mathbf{x} * \mathbf{y} = \sum_{j=1}^n \left( x_j y_j^2 + x_j^2 y_j\right).
\end{equation}
Given an additive code $C$, its {\it symplectic dual} $C^*$ is 
\[
C^* = \{ \mathbf{x} \in \mathbb{F}_4^n : \mathbf{x} * \mathbf{c} =0 \mbox{ for all } \mathbf{c} \in C\}
\]
and $C$ is said to be {\it (symplectic) self-dual} if $C = C^*$.

An additive $\mathbb{F}_4$ self-dual code is called {\it Type II} if all of its codewords have even weights. A code which is not Type II is called {\it Type I}. It is well-known that Type II codes must have even lengths.

Self-dual codes, under various inner products and possible alphabet sets, have been extensively studied due to their rich algebraic, combinatorial, and geometric structures. A major reference on this topic is the book \cite{NRS06} authored by Nebe, Rains, and Sloane. Of particular relevance to our family of additive self-dual codes here, labeled as family $4^{H+}$ in \cite{NRS06}, is the treatment in Section 6 of Chapter 7 and in Chapter 11. 

Let $A(\Gamma)$ be the adjacency matrix of the graph $\Gamma$ and $I$ be the identity matrix. The following nice result was first shown by Schlingemann in~\cite{Schlingemann2002} and subsequently discussed in \cite[Section~3]{Danielsen2006}. Every graph represents a self-dual additive code over $\mathbb{F}_4$ and every self-dual additive code over $\mathbb{F}_4$ can be represented by a graph. In particular, the additive $\mathbb{F}_4$-code $C:=C(\Gamma)$ generated by the row span of the matrix $A(\Gamma) + \omega I$ is symplectic self-dual. 

Danielsen and Parker gave a complete classification of \emph{all} self-dual additive codes over $\mathbb{F}_4$ for $n \leq 12$ in~\cite{Danielsen2006}. Follow-up works, covering $n \leq 50$, were contributed by Gulliver and Kim in~\cite{GK2004}, by Varbanov in~\cite{Varbanov2008}, by Grassl and Harada in~\cite{Grassl2017}, and by Saito in~\cite{Saito2019}. Their collective efforts focused on codes derived from graphs whose adjacency matrices are either {\it circulant} or {\it bordered circulant}.

The following result classifies Type I and Type II additive $\mathbb{F}_4$ self-dual codes generated by metacirculant graphs. 

\begin{theorem}\label{thm:even}
Let $C$ be an additive $\mathbb{F}_4$ self-dual code generated by 
\[
\Gamma\left(m,n,\alpha, S_{0}, S_{1}, \ldots, S_{\floor{m/2}}\right)\mbox{, with }
2 \mid (mn).
\]
Let
\begin{equation}\label{eq:even}
\Delta_S :=
\begin{cases} 
	|S_0|  & \mbox{if } m \mbox{ is odd},\\
	|S_0| + |S_{\floor{m/2}}|  & \mbox{if } m \mbox{ is even}.
\end{cases}
\end{equation}
Then $C$ is Type II if and only if $\Delta_S$ is odd.
\end{theorem}

\begin{proof}
Let $mn$ be even. An additive self-dual $\mathbb{F}_4$ code $C = C(\Gamma)$ is Type II if and only if all vertices of $\Gamma$ have odd degree~\cite{Danielsen2006}. Lemma~\ref{lem:valence} gives the valence of $\Gamma\left(m,n,\alpha,S_{0}, S_{1}, \ldots, S_{\floor{m/2}}\right)$. By Equation~\ref{eq:valence}, the valence is odd if and only if $\Delta_S$ is odd.
\end{proof}

The three parameters of a qubit code $Q \subseteq \mathbb{C}^{2^{\ell}}$ are its \textit{length} $\ell$, \textit{dimension} $K$ over $\mathbb{C}$, and \textit{minimum distance} $d=d(Q)$. The notation 
\[
((\ell,K,d)) \mbox{ or } \dsb{\ell,k,d} \mbox{ with } k = \log_2 K
\] 
signifies that $Q$ encodes $K$ logical qubits as $\ell$ physical qubits, with $d$ being the smallest number of simultaneous quantum error operators that can send a valid codeword into another. 

A symplectic self-dual additive code $C$ over $\mathbb{F}_4$ of length $\ell$ and minimum distance $d$ gives an $\dsb{\ell,0,d}_2$ qubit code $Q$. Since $k=0$, that is the code $Q$ consists of a single quantum state, one needs to carefully interpret the meaning of minimum distance. As explained in \cite[Sect. III]{Calderbank1998}, an $\dsb{\ell,0,d}$ code $Q$ has the property that, when
subjected to a decoherence of $\floor{(d-1)/2}$ coordinates, it is possible
to determine exactly which coordinates were decohered. This code can be used, for example, to test if certain storage locations for qubits are decohering faster than
they should.

Such a code $Q$ can be of interest in their own right. A famous example is the unique $\dsb{2,0,2}$ code that corresponds to the maximally entangled quantum state known as the EPR pair in the famed paper \cite{EPR} of Einstein, Podolsky, and Rosen. More commonly, a zero-dimensional code is used as a seed in some secondary constructions of quantum codes to produce qubit codes with $k > 0$.

\begin{theorem}\cite[Theorem 6]{Calderbank1998}
Assume that a qubit $((\ell, K, d>1 ))_2$ code $Q$ exists. Then the following qubit codes exist. An $((\ell, K', d))_2$ code for all $1 < K' \leq K$ by {\it subcode construction}. A $((\lambda, K, d))_2$ code for all $K > 1$ and $\lambda \geq \ell$ by {\it lengthening}. An $((\ell-1, K, d-1))_2$ code by {\it puncturing}.
\end{theorem}
There is also a quantum analogue of the {\it shortening} construction on classical code, although the former is less straightforward to perform. Interested reader can consult \cite[Section 4.3]{Grassl20} for the procedure.

\begin{example}
The $[12,6,6]_4$ {\tt dodecacode} $\mathcal{D}$ yields the unique $\dsb{12, 0, 6}$ qubit code. It can be generated by the metacirculant graph $G_{12}:=\Gamma(2, 6, 5, \{ 3 \},\{ 0, 3, 4, 5 \})$. Figure~\ref{fig:12} shows $G_{12}$ with the vertices relabeled for convenience. The {\tt dodecacode} is Type II, with weight distribution
		
\begin{center}
\begin{tabular}{c r |c r | c r | c r | c r}
\hline
${\rm wt}$ & $\# \mathbf{c}$ & 
${\rm wt}$ & $\# \mathbf{c}$ &
${\rm wt}$ & $\# \mathbf{c}$ & 
${\rm wt}$ & $\# \mathbf{c}$ &
${\rm wt}$ & $\# \mathbf{c}$\\ 
\hline
$0$ & $1$ & $6$ & $396$ & $8$ & $1,485$ 
& $10$ & $1,980$ & $12$ & $234$ \\
\hline
\end{tabular}.
\end{center}
	
\begin{figure}[h!]
\centering
\begin{tikzpicture}[multilayer=3d]
			\SetLayerDistance{-2.5}
			\Plane[x=-.5,y=-.5,width=6,height=3,color=gray,layer=2,NoBorder]
			\Plane[x=-.5,y=-.5,width=6,height=3,NoBorder]
			\Vertex[x=0.1,y=0.2,IdAsLabel,layer=1]{1}
			\Vertex[x=2.4,y=0.2,IdAsLabel,layer=1]{2}
			\Vertex[x=5.0,y=0.2,IdAsLabel,layer=1]{3}
			\Vertex[x=0.1,y=2.0,IdAsLabel,layer=1]{6}
			\Vertex[x=2.4,y=2.0,IdAsLabel,layer=1]{5}
			\Vertex[x=5.0,y=2.0,IdAsLabel,layer=1]{4}
			% second layer	
			\Vertex[x=0.1,y=0.2,color=red,IdAsLabel,layer=2]{10}
			\Vertex[x=2.4,y=0.0,color=red,IdAsLabel,layer=2]{11}
			\Vertex[x=5.0,y=0.2,color=red,IdAsLabel,layer=2]{12}
			\Vertex[x=-0.1,y=2.0,color=red,IdAsLabel,layer=2]{9}
			\Vertex[x=2.4,y=1.8,color=red,IdAsLabel,layer=2]{8}
			\Vertex[x=5.0,y=2.0,color=red,IdAsLabel,layer=2]{7}
			
			%% edge list layer 1
			\Edge(1)(2)
			\Edge(1)(6)
			\Edge(2)(5) 
			\Edge(3)(4)
			\Edge(5)(6)
			%%% edge list layer 2
			\Edge(11)(12)
			\Edge[style=dashed](8)(11)
			\Edge[style=dashed](7)(8)
			\Edge(7)(12)
			\Edge[style=dashed](9)(10)
			%%% edges between layers
			\Edge[color=blue](1)(10)
			\Edge[style=dashed,bend=-40,color=blue](1)(7)
			\Edge[style=dashed,bend=5,color=blue](1)(8)
			
			\Edge[style=dashed,bend=-20,color=blue](2)(9)
			\Edge[style=dashed,color=blue](2)(7)
			\Edge[style=dashed,color=blue](2)(8)
			
			\Edge[style=dashed,color=blue](3)(9)
			\Edge[style=dashed,color=blue](3)(10)
			\Edge[style=dashed,color=blue](3)(8)
			\Edge[color=blue](3)(12)
			
			\Edge[style=dashed,color=blue](4)(9)
			\Edge[style=dashed,bend=-5,color=blue](4)(10)
			\Edge[color=blue](4)(7)
			\Edge[bend=30,color=blue](4)(11)
			
			\Edge[style=dashed,bend=-30,color=blue](5)(10)
			\Edge[style=dashed,color=blue](5)(11)
			\Edge[style=dashed,color=blue](5)(12)
			
			\Edge[style=dashed,bend=-5,color=blue](6)(9)
			\Edge[style=dashed,color=blue](6)(11)
			\Edge[bend=8,style=dashed,color=blue](6)(12)
			
\end{tikzpicture}
\caption{$G_{12}:=\Gamma(2, 6, 5, \{ 3 \},\{ 0, 3, 4, 5 \})$ of the {\tt dodecacode} $\mathcal{D}$. \quad \QEDB}
\label{fig:12}
\end{figure}
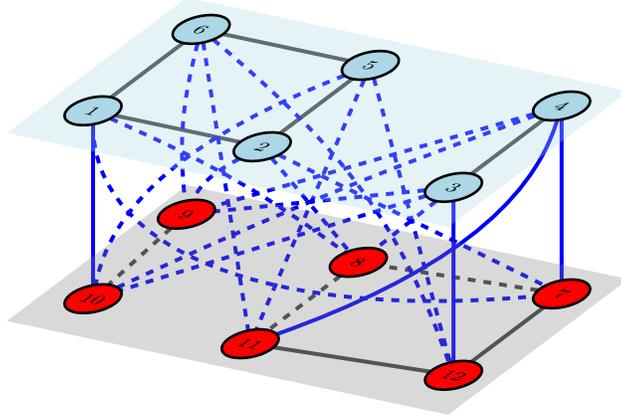
\end{example}

Many known optimal or currently best-performing qubit codes of length $\ell$ with $K=1$ in the literature, \ie, $\dsb{\ell,0,d}_2$ with best-known or optimal $d$, are constructed based on circulant graphs \cite{Danielsen2006,Grassl2017,Saito2019}. We will soon show that strict improvements can be gained when one starts with metacirculant graphs. We use {\tt MAGMA}~\cite{BOSMA1997} for all computations.

%%%%%%%%% New Section

\section{Higher best-known minimum distances of additive symplectic self-dual codes}\label{sec:betterSD}

This section presents additive symplectic self-dual codes of length $\ell \in \{27,36\}$ with stricly higher minimum distances than any that can be derived from circulant graphs. The strict improvements, however, do not extend to the quantum setup. There are $\dsb{\ell,0,d}$ qubit codes, constructed from other approaches as recorded in the corresponding entries in \cite{Grassl:codetables}, whose minimum distances are equal to the ones that we derive here.

For length $\ell=27$, Varbanov in~\cite{Varbanov2008} concluded after an exhaustive search that the circulant graph construction only yields additive self-dual codes $\left(27,2^{27},\widehat{d}\right)_4$ of highest minimum distance $d = 8 \geq \widehat{d}$. Using the metacirculant graph construction, we obtain additive self-dual codes with parameters $\left(27,2^{27},9\right)_4$. They yield $\dsb{27,0,9}$ qubit codes, which meet the best-known parameters in \cite{Grassl:codetables}. 

%Here is the graph for 272
\begin{figure}[h!]
	\centering
	\begin{tikzpicture}[multilayer=3d]
		\SetLayerDistance{-4.5}
		\Plane[x=-.4,y=-.4,width=10,height=5,color=gray,layer=2,NoBorder]
		\Plane[x=-.4,y=-.4,width=10,height=5,color=green,layer=3,NoBorder]
		\Plane[x=-.4,y=-.4,width=10,height=5,NoBorder]
		\Vertex[x=1.5,y=0.3,IdAsLabel,layer=1]{8}
		\Vertex[x=4.7,y=0.3,IdAsLabel,layer=1]{7}
		\Vertex[x=7.5,y=0.3,IdAsLabel,layer=1]{6}
		\Vertex[x=0.3,y=2.2,IdAsLabel,layer=1]{1}
		\Vertex[x=4.5,y=2.2,IdAsLabel,layer=1]{9}
		\Vertex[x=9.0,y=2.3,IdAsLabel,layer=1]{5}
		\Vertex[x=1.3,y=4.0,IdAsLabel,layer=1]{2}
		\Vertex[x=4.5,y=4.0,IdAsLabel,layer=1]{3}
		\Vertex[x=8.0,y=4.0,IdAsLabel,layer=1]{4}
		% second layer	
		\Vertex[x=0.5,y=0.3,IdAsLabel,layer=2]{17}
		\Vertex[x=3.2,y=0.3,IdAsLabel,layer=2]{11}
		\Vertex[x=5.9,y=0.3,IdAsLabel,layer=2]{12}
		\Vertex[x=9.0,y=0.3,IdAsLabel,layer=2]{16}
		\Vertex[x=4.5,y=2.2,IdAsLabel,layer=2]{10}
		\Vertex[x=0.5,y=4.0,IdAsLabel,layer=2]{15}
		\Vertex[x=3.2,y=4.0,IdAsLabel,layer=2]{18}
		\Vertex[x=5.9,y=4.0,IdAsLabel,layer=2]{13}
		\Vertex[x=9.0,y=4.0,IdAsLabel,layer=2]{14}
		% third layer	
		\Vertex[x=1.5,y=0.3,IdAsLabel,layer=3]{25}
		\Vertex[x=4.5,y=0.3,IdAsLabel,layer=3]{24}
		\Vertex[x=7.8,y=0.5,IdAsLabel,layer=3]{22}
		\Vertex[x=0.3,y=2.2,IdAsLabel,layer=3]{19}
		\Vertex[x=4.5,y=2.2,IdAsLabel,layer=3]{20}
		\Vertex[x=9.0,y=2.3,IdAsLabel,layer=3]{21}
		\Vertex[x=1.3,y=4.2,IdAsLabel,layer=3]{26}
		\Vertex[x=4.5,y=4.0,IdAsLabel,layer=3]{23}
		\Vertex[x=8.0,y=4.0,IdAsLabel,layer=3]{27}
		%% edge list layer 1
		\Edge(1)(2)\Edge(1)(4)\Edge(1)(7)\Edge(1)(9)
		\Edge(2)(3)\Edge(2)(5)\Edge(2)(8)
		\Edge(3)(4)\Edge(3)(6)\Edge(3)(9)
		\Edge(4)(5)\Edge(4)(7)
		\Edge(5)(6)\Edge(5)(8)
		\Edge(6)(7)\Edge(6)(9)
		\Edge(7)(8)\Edge(8)(9)
		%%% edge list layer 2
		\Edge[style=dashed](10)(12)\Edge[style=dashed,bend=-45](10)(13)
		\Edge[style=dashed,bend=5](10)(16)\Edge[style=dashed](10)(17)
		\Edge[style=dashed,bend=25](11)(13)\Edge[style=dashed](11)(14)
		\Edge(11)(17)\Edge[style=dashed](11)(18)
		\Edge[style=dashed](12)(14)\Edge[style=dashed](12)(15)
		\Edge[style=dashed,bend=-18](12)(18)
		\Edge[style=dashed,bend=20](13)(15)\Edge[style=dashed](13)(16)
		\Edge(14)(16)\Edge[style=dashed,bend=-15](14)(17)
		\Edge[style=dashed](15)(17)\Edge[style=dashed](15)(18)
		\Edge[style=dashed](16)(18)
		%%% edge list layer 3
		\Edge[style=dashed](19)(22)\Edge[style=dashed](19)(23)
		\Edge[style=dashed](19)(24)\Edge[style=dashed](19)(25)
		\Edge[bend=-40,style=dashed](20)(23)
		\Edge[bend=-30,style=dashed](20)(24)
		\Edge[style=dashed](20)(25)\Edge[style=dashed](20)(26)
		\Edge[style=dashed](21)(24)\Edge[style=dashed](21)(25)
		\Edge[style=dashed](21)(26)\Edge[style=dashed](21)(27)
		\Edge[bend=15](22)(25)\Edge[bend=-15,style=dashed](22)(26)
		\Edge[style=dashed](22)(27)
		\Edge[style=dashed](23)(26)\Edge[style=dashed](23)(27)
		\Edge[style=dashed](24)(27)
		%%% edges among layers
		%% layer 1 to layer 2
		\Edge[style=dashed,color=red](1)(10)
		\Edge[color=red,style=dashed](1)(15)
		\Edge[bend=-5,style=dashed,color=red](1)(18)
		
		\Edge[color=red,bend=-12,style=dashed](2)(10)
		\Edge[bend=-15,color=red](2)(11)
		\Edge[color=red](2)(16)
		
		\Edge[bend=10,color=red](3)(11)
		\Edge[bend=18,color=red](3)(12)
		\Edge[bend=-40,color=red](3)(17)
		
		\Edge[color=red,bend=45](4)(12)
		\Edge[bend=30,style=dashed,color=red](4)(13)
		\Edge[bend=15,style=dashed,color=red](4)(18)
		
		\Edge[style=dashed,color=red](5)(10)
		\Edge[style=dashed,color=red](5)(13)
		\Edge[color=red](5)(14)
		
		\Edge[bend=-40,color=red](6)(11)
		\Edge[color=red](6)(14)
		\Edge[style=dashed,color=red](6)(15)
		
		\Edge[bend=-5,color=red](7)(12)
		\Edge[style=dashed,color=red](7)(15)
		\Edge[bend=10,color=red](7)(16)
		
		\Edge[style=dashed,color=red](8)(13)
		\Edge[bend=-15,color=red](8)(16)
		\Edge[color=red](8)(17)
		
		\Edge[style=dashed,color=red](9)(14)
		\Edge[style=dashed,bend=-25,color=red](9)(17)
		\Edge[style=dashed,color=red](9)(18)
		%% layer 2 to layer 3
		\Edge[style=dashed,color=blue](10)(19)
		\Edge[style=dashed,color=blue](10)(21)
		\Edge[style=dashed,bend=20,color=blue](10)(27)
		
		\Edge[color=blue](11)(19)
		\Edge[bend=-15,style=dashed,color=blue](11)(20)
		\Edge[bend=15,color=blue](11)(22)
		
		\Edge[style=dashed,color=blue](12)(20)
		\Edge[bend=5,color=blue](12)(21)
		\Edge[style=dashed,color=blue](12)(23)
		
		\Edge[style=dashed,color=blue](13)(21)
		\Edge[style=dashed,bend=25,color=blue](13)(22)
		\Edge[style=dashed,bend=-10,color=blue](13)(24)
		
		\Edge[bend=15,color=blue](14)(22)
		\Edge[style=dashed,bend=35,color=blue](14)(23)
		\Edge[bend=-5,color=blue](14)(25)
		
		\Edge[style=dashed,bend=-10,color=blue](15)(23)
		\Edge[style=dashed,bend=-2,color=blue](15)(24)
		\Edge[style=dashed,color=blue](15)(26)
		
		\Edge[bend=30,color=blue](16)(24)
		\Edge[color=blue](16)(25)
		\Edge[color=blue](16)(27)
		
		\Edge[color=blue](17)(19)
		\Edge[bend=-5,color=blue](17)(25)
		\Edge[style=dashed,bend=-15,color=blue](17)(26)
		
		\Edge[style=dashed,bend=10,color=blue](18)(20)
		\Edge[style=dashed,bend=-10,color=blue](18)(26)
		\Edge[style=dashed,bend=20,color=blue](18)(27)
		
		%% layer 1 to layer 3
		\Edge[bend=5,style=dashed,color=black](1)(19)
		\Edge[bend=2,style=dashed,color=black](1)(23)
		\Edge[bend=0,style=dashed,color=black](1)(26)
		
		\Edge[bend=-8,style=dashed,color=black](2)(20)
		\Edge[bend=-20,color=black](2)(24)
		\Edge[style=dashed,color=black](2)(27)
		
		\Edge[bend=5,style=dashed,color=black](3)(19)
		\Edge[bend=4,color=black](3)(21)
		\Edge[bend=6,color=black](3)(25)
		
		\Edge[bend=-7,style=dashed,color=black](4)(20)
		\Edge[bend=15,color=black](4)(22)
		\Edge[bend=47,style=dashed,color=black](4)(26)
		
		\Edge[bend=-7,color=black](5)(21)
		\Edge[bend=2,style=dashed,color=black](5)(23)
		\Edge[bend=5,color=black](5)(27)
		
		\Edge[bend=-45,style=dashed,color=black](6)(19)
		\Edge[bend=6,color=black](6)(22)
		\Edge[bend=10,color=black](6)(24)
		
		\Edge[bend=-2,style=dashed,color=black](7)(20)
		\Edge[bend=25,style=dashed,color=black](7)(23)
		\Edge[bend=8,color=black](7)(25)
		
		\Edge[bend=20,color=black](8)(21)
		\Edge[bend=-12,color=black](8)(24)
		\Edge[bend=-15,style=dashed,color=black](8)(26)
		
		\Edge[bend=8,color=black](9)(22)
		\Edge[bend=8,color=black](9)(25)
		\Edge[bend=15,color=black](9)(27)
	\end{tikzpicture}
\caption{$G_{27,2}:=\Gamma(3, 9, 4, [\{ 2, 7 \}, \{ 0, 1, 2, 3, 4, 5, 8 \}])$ from the family $\mathcal{F}_{27,2}$. Edges between $\mathcal{L}_1$ and $\mathcal{L}_2$ are in red, between $\mathcal{L}_1$ and $\mathcal{L}_3$ are in black, between $\mathcal{L}_2$ and $\mathcal{L}_3$ are in blue.}
\label{fig:272}
\end{figure}
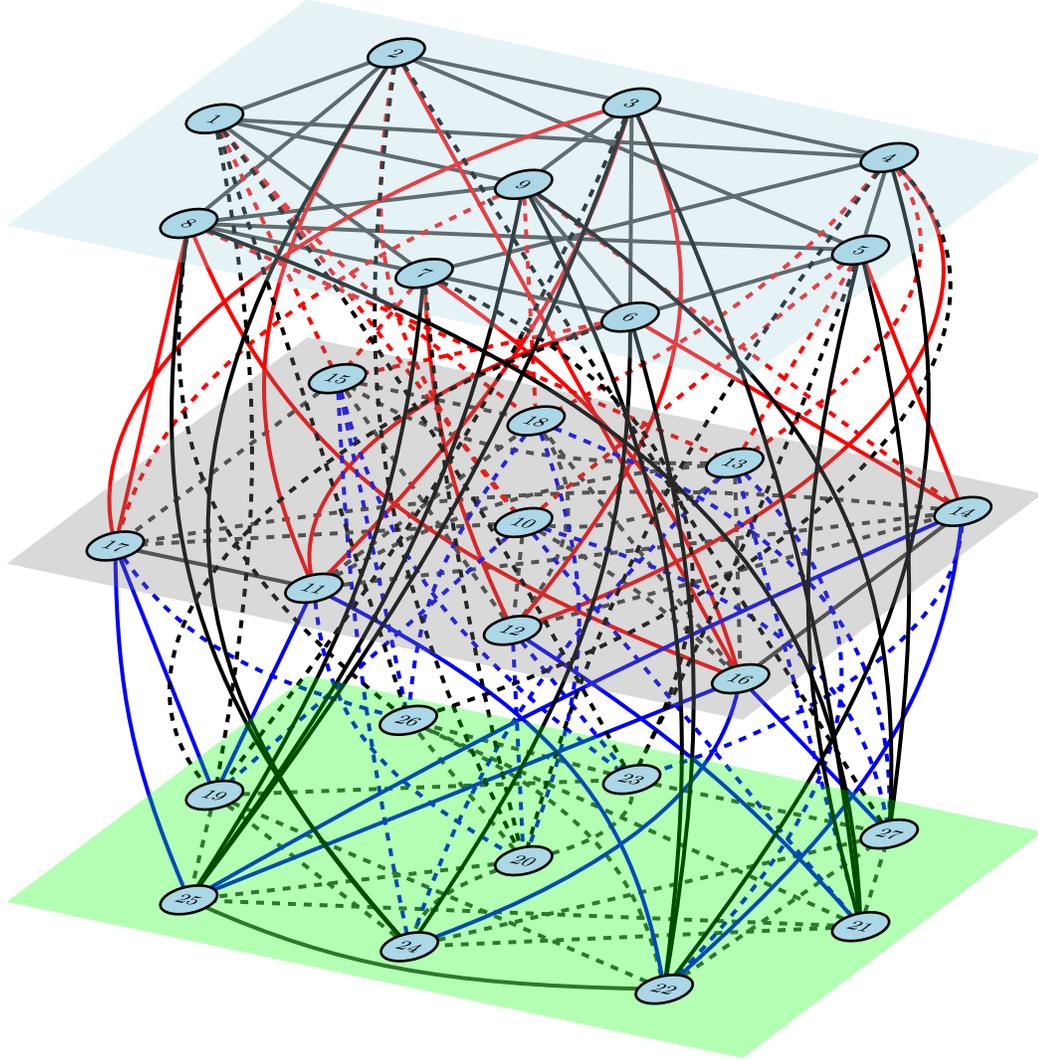
%%% here above is graph of 272

\begin{proposition}\label{prop27}
There are two inequivalent families of $\left(27,2^{27},9\right)_4$ additive self-dual codes, each consisting of $108$ equivalent codes. They yield $\dsb{27,0,9}$ qubit stabilizer codes.
\end{proposition}

\begin{proof} 
An exhaustive search found $216$ metacirculant graphs 
$\{\Gamma(3, 9, \alpha, S_0, S_1)\}$ 
in two non-isomorphic families, corresponding to two inequivalent codes. Tables~\ref{table:list271} and \ref{table:list272} give the complete list of $\alpha, S_0, S_1$.

The families $\mathcal{F}_{27,1}$ and $\mathcal{F}_{27,2}$ are represented by the graphs 
\begin{align}
	G_{27,1} &:=\Gamma(3, 9, 4, \{ 2, 7 \}, \{ 0, 1, 2, 3, 4, 5, 8 \}) \mbox{ and } \notag\\
	G_{27,2} &:=\Gamma(3, 9, 7, \{ 1, 3, 6, 8 \},\{ 0, 5, 8 \}).
\end{align}
Figure~\ref{fig:272} visualizes $G_{27,2}$ according to the vertex relabeling done in Table~\ref{graph:27} to highlight its metacirculant structure. One can do similarly for the more complicated looking $G_{27,1}$. The even weights in both families of codes have the same number of codewords. The weight distributions are in Table~\ref{table:27}.   
\end{proof}

%%%% List of edges
\begin{table}
\caption{List of edges of $G_{27,1}$ and in $G_{27,2}$. 
Vertices $1$ to $9$ are in $\mathcal{L}_1$. Vertices $10$ to $18$ are in $\mathcal{L}_2$. Vertices $19$ to $27$ are in $\mathcal{L}_3$. Edges $\{(i,j)\}$ are listed as $(i,\{j \in J\})$ with $i \in \mathcal{L}_r$ and $J \subset \mathcal{L}_s$}
\centering
\small
\begin{tabular}{ccl}
		\toprule
		\multicolumn{3}{c}{The $216$ Edges of $G_{27,1}$} \\
		\midrule
		$\#$ & $(r,s)$ & \multicolumn{1}{c}{$(i,\{j \in J\})$}\\
		\midrule
		$9$ & $(1,1)$ & $(1, \{3,8\}), \,(2, \{4,9\}), \,(3, \{5\}), \,(4, \{6\}), \,(5, \{7\}), \,
		(6, \{8\}), \,(7, \{9\}) $ \\
		
		& $(2,2)$ & 
		$(10, \{11,18\}), \,(11, \{12\}), \,(12, \{13\}), \,(13, \{14\}), \,
		(14, \{15\})$ \\
		&& $(15, \{16\}), \,(16, \{17\}), \,(17, \{18\})$ \\
		& $(3,3)$ & $(19, \{23,24\}), \,(20, \{24,25\}), \,(21, \{25,26\}), \,
		(22, \{26,27\}), \, (23, \{27\})$ \\
		
		$63$ & $(1,2)$ & $(1, \{10,11,12,13,14,15,18\}), \,
		(2, \{10,11,12,13,14,15,16\}),$ \\
		
		& & $(3, \{11,12,13,14,15,16,17\}), \,(4, \{12,13,14,15,16,17,18\}),$ \\
		
		&& $(5, \{10,13,14,15,16,17,18\}), \,(6, \{10,11,14,15,16,17,18\}),$ \\
		
		&& $(7, \{10,11,12,15,16,17,18\}), \,(8, \{10,11,12,13,16,17,18\}),$ \\ 
		
		&& $(9, \{10,11,12,13,14,17,18\})$ \\
		
		$63$ & $(1,3)$ & $(1, \{19,20,21,23,25,26,27\}), \,(2, \{19,20,21,22,24,26,27\}),$ \\
		
		&& $(3, \{19,20,21,22,23,25,27\}), \,(4, \{19,20,21,22,23,24,26\}),$ \\
		
		&& $(5, \{20,21,22,23,24,25,27\}), \,(6, \{19,21,22,23,24,25,26\}),$ \\
		
		&& $(7, \{20,22,23,24,25,26,27\}), \,(8, \{19,21,23,24,25,26,27\}),$ \\
		
		&& $(9, \{19,20,22,24,25,26,27\})$ \\
		
		$63$ & $(2,3)$ & 
		$(10, \{19,21,22,23,24,26,27\}), \,(11, \{19,20,22,23,24,25,27\}),$ \\
		
		&& $(12, \{19,20,21,23,24,25,26\}), \,(13, \{20,21,22,24,25,26,27\}),$\\
		
		&& $(14, \{19,21,22,23,25,26,27\}), \,(15, \{19,20,22,23,24,26,27\}),$ \\
		
		&& $(16, \{19,20,21,23,24,25,27\}), \,(17, \{19,20,21,22,24,25,26\}),$ \\
		
		&& $(18, \{20,21,22,23,25,26,27\})$ \\
		\toprule 
		\multicolumn{3}{c}{The $135$ Edges of $G_{27,2}$}\\
		\midrule
		$\#$ & $(r,s)$ & \multicolumn{1}{c}{$(i,\{j \in J\})$} \\
		\midrule
		$18$ & $(1,1)$ & $(1,\{2,4,7,9\}), \,(2,\{3,5,8\}), \,(3,\{4,6,9\}), \,(4,\{5,7\}),$\\
		& & $(5,\{6,8\}), \,(6,\{7,9\}), \,(7,\{8\}), \,(8,\{9\})$ \\
		
		$18$ & $(2,2)$ & $(10,\{12,13,16,17\}), \,(11,\{13,14,17,18\}), \,
		(12,\{14,15,18\})$\\
		
		&& $(13,\{15,16\}), \,(14,\{16,17\}), \,(15,\{17,18\}), \,(16,\{18\})$\\
		
		$18$ & $(3,3)$ & $(19,\{22,23,24,25\}), \,(20,\{23,24,25,26\}), \,(21,\{24,25,26,27\}),$\\
		
		&& $(22,\{25,26,27\}), \,(23,\{26,27\}), \,(24,\{27\})$ \\
		
		$27$ & $(1,2)$ & $(1,\{10,15,18\}), \,(2,\{10,11,16\}), \,(3,\{11,12,17\}),$ \\
		
		&& $(4,\{12,13,18\}), \,(5,\{10,13,14\}), \,(6,\{11,14,15\}),$\\
		
		&& $(7,\{12,15,16\}), \,(8,\{13,16,17\}), \,(9,\{14,17,18\})$\\ 
		
		$27$ & $(1,3)$ & $(1,\{19,23,26\}), \,(2,\{20,24,27\}), \,(3,\{19,21,25\}),$\\
		
		&& $(4,\{20,22,26\}), \,(5,\{21,23,27\}), \,(6,\{19,22,24\}),$\\
		&& $(7,\{20,23,25\}), \,(8,\{21,24,26\}), \,(9,\{22,25,27\})$\\
		
		$27$ & $(2,3)$ & $(10,\{19,21,27\}), \,(11,\{19,20,22\}), \,(12,\{20,21,23\}),$\\
		&& $(13,\{21,22,24\}), \,(14,\{22,23,25\}), \,(15,\{23,24,26\}),$\\
		&& $(16,\{24,25,27\}), \,(17,\{19,25,26\}), \,(18,\{20,26,27\})$\\
		\bottomrule 
	\end{tabular} 
	\label{graph:27}
\end{table}
%%% end of list of edges

%%% Table for Family 27_1 in length 27
\begin{table}
\caption{List of $\alpha$ and $S_0,S_1$ for $\Gamma(3,9,\alpha,S_0,S_1)$ corresponding to $\mathcal{F}_{27,1}$ in Proposition~\ref{prop27}.}
\centering
\small
\begin{tabular}{c c l}
\toprule
\multicolumn{3}{c}{Family $\mathcal{F}_{27,1}$} \\
\midrule
$\alpha$ & $S_0$ & \multicolumn{1}{c}{$\{S_1\}$} \\ 
\midrule
$4$ & $\{2,7\}$ & $\{\{ 0, 1, 2, 3, 4, 5, 8 \},~\{ 0, 1, 2, 5, 6, 7, 8 \},~
\{ 0, 1, 4, 5, 6, 7, 8 \}$,\\
&& $~\{ 0, 1, 2, 3, 4, 7, 8 \},~\{ 2, 3, 4, 5, 6, 7, 8 \},
~\{ 1, 2, 3, 4, 5, 6, 7 \}\}$. \\
		
& $\{4,5\}$ & $\{\{ 0, 1, 2, 3, 4, 5, 7 \},~\{ 0, 2, 4, 5, 6, 7, 8 \},~
\{ 1, 2, 3, 4, 5, 6, 8 \}$,\\
&&$~\{ 1, 3, 4, 5, 6, 7, 8 \},~ 
\{ 0, 1, 2, 4, 6, 7, 8 \},~\{ 0, 1, 2, 3, 5, 7, 8 \}\}$.\\
		
& $\{1,8\}$ & $\{\{ 0, 1, 2, 4, 5, 6, 8 \},~\{ 1, 2, 3, 5, 6, 7, 8 \},~
\{ 0, 1, 2, 4, 5, 6, 7 \}$,\\
&&$~\{ 1, 2, 3, 4, 6, 7, 8 \},~
\{ 0, 2, 3, 4, 5, 7, 8 \},~\{ 0, 1, 3, 4, 5, 7, 8 \}\}$.\\
		
& $\{ 1, 2, 4, 5, 7, 8 \}$ & 
		$\{\{ 0, 2, 4, 6, 7 \},~\{ 0, 3, 4, 5, 8 \},~\{ 0, 1, 3, 5, 7 \}$,\\
		&&$~\{ 0, 1, 2, 3, 4 \},~\{ 1, 3, 5, 6, 8 \},~\{ 1, 2, 3, 6, 7 \}\}$.\\
		
		&&$\{\{ 0, 1, 2, 3, 8 \},~\{ 3, 4, 5, 6, 7 \}~\{ 0, 2, 3, 5, 7 \}$,\\
		&&$~\{ 1, 3, 4, 6, 8 \},~\{ 0, 3, 4, 7, 8 \}, ~\{ 2, 3, 4, 5, 6 \}\}$. \\
		
		& & $\{\{ 0, 1, 6, 7, 8 \},~\{ 2, 3, 6, 7, 8 \},~\{ 0, 1, 2, 5, 6 \}$,\\
		&& $~\{ 0, 2, 4, 6, 8 \},~\{ 0, 1, 4, 5, 6 \},~\{ 0, 5, 6, 7, 8 \}\}$.\\
		
		& $\{ 1, 2, 7, 8 \}$ & $\{\{ 0, 1, 2, 3, 5, 8 \},~\{0,1, 4, 6, 7, 8 \},~\{ 1, 3, 4, 5, 6, 7 \}$,\\
		&&$~\{ 0, 2, 5, 6, 7, 8 \},~\{ 2, 3, 4, 5, 6, 8\},~\{ 0, 1, 2, 3, 4, 7 \}\}$.\\

		& $\{2, 4, 5, 7 \}$ &
		$\{\{ 0, 1, 2, 4, 6, 7 \},~\{ 0, 2, 3, 5, 7, 8 \},~
		\{ 1, 3, 4, 6, 7, 8 \}$,\\
		&& $~\{ 0, 2, 4, 5, 6, 8 \},~
		\{ 1, 2, 3, 5, 6, 8 \},~\{ 0, 1, 3, 4, 5, 7 \}\}$. \\
		
		& $\{ 1, 4, 5, 8 \}$ & 
		$\{\{ 0, 1, 4, 5, 6, 7 \},~\{ 2, 3, 5, 6, 7, 8 \},~
		\{ 0, 2, 3, 4, 5, 8 \}$,\\
		&&$~\{ 1, 2, 3, 4, 6, 7 \},~
		\{ 0, 1, 2, 5, 6, 8 \},~\{ 0, 1, 3, 4, 7, 8 \}\}$.\\
		
\midrule
		
$7$ & $\{2,7\}$ & $\{\{ 0, 1, 2, 3, 4, 5, 7 \},~
\{ 0, 2, 4, 5, 6, 7, 8 \},~\{ 1, 2, 3, 4, 5, 6, 8 \}$,\\
&&$~\{ 1, 3, 4, 5, 6, 7, 8 \},~\{ 0, 1, 2, 4, 6, 7, 8 \},~\{ 0, 1, 2, 3, 5, 7, 8 \}\}$. \\
		
& $\{4,5\}$ & $\{\{ 0, 1, 2, 4, 5, 6, 8 \},~
\{ 1, 2, 3, 5, 6, 7, 8 \},~
\{ 0, 1, 2, 4, 5, 6, 7 \}$,\\
&&$~\{ 1, 2, 3, 4, 6, 7, 8 \},~
		\{ 0, 2, 3, 4, 5, 7, 8 \},~
		\{ 0, 1, 3, 4, 5, 7, 8 \}\}$.\\
		
		& $\{1,8\}$ & $\{\{0, 1, 2, 3, 4, 5, 8\},~
		\{0, 1, 2, 5, 6, 7, 8\},~
		\{0, 1, 4, 5, 6, 7, 8\}$,\\
		&&$~\{0, 1, 2, 3, 4, 7, 8\},~
		\{2, 3, 4, 5, 6, 7, 8\},~
		\{1, 2, 3, 4, 5, 6, 7\}\}$.\\
		
		& $\{ 1, 2, 4, 5, 7, 8 \}$ & 
		$\{\{0, 2, 4, 6, 7\},~
		\{0, 3, 4, 5, 8\},~
		\{0, 1, 3, 5, 7\}$,\\
		&&$~\{0, 1, 2, 3, 4\},~
		\{1, 3, 5, 6, 8\},~
		\{1, 2, 3, 6, 7\}\}$. \\ 
		
		& & $\{\{0, 1, 2, 3, 8\},~
		\{3, 4, 5, 6, 7\},~
		\{0, 2, 3, 5, 7\}$,\\
		&&$~\{1, 3, 4, 6, 8\},~
		\{0, 3, 4, 7, 8\},~
		\{2, 3, 4, 5, 6\}\}$. \\
		
		& & $\{\{0, 1, 6, 7, 8\},~
		\{2, 3, 6, 7, 8\},~
		\{0, 1, 2, 5, 6\}$,\\
		&&$~\{0, 2, 4, 6, 8\},~
		\{0, 1, 4, 5, 6\},~
		\{0, 5, 6, 7, 8\}\}$.\\
		
		& $\{ 1, 2, 7, 8 \}$ & $\{\{0, 1, 2, 4, 6, 7\},~
		\{0, 2, 3, 5, 7, 8\},~
		\{1, 3, 4, 6, 7, 8\}$,\\
		&&$~\{0, 2, 4, 5, 6, 8\},~
		\{1, 2, 3, 5, 6, 8\},~
		\{0, 1, 3, 4, 5, 7\}\}$.\\
		
		& $\{2, 4, 5, 7 \}$ &
		$\{\{0, 1, 4, 5, 6, 7\},~
		\{2, 3, 5, 6, 7, 8\},~
		\{0, 2, 3, 4, 5, 8\}$,\\
		&&$~\{1, 2, 3, 4, 6, 7\},~
		\{0, 1, 2, 5, 6, 8\},~
		\{0, 1, 3, 4, 7, 8\}\}$. \\
		
		& $\{ 1, 4, 5, 8 \}$ & 
		$\{\{0, 1, 2, 3, 5, 8\},~
		\{0, 1, 4, 6, 7, 8\},~
		\{1, 3, 4, 5, 6, 7\}$,\\
		&&$~\{0, 2, 5, 6, 7, 8\},~
		\{2, 3, 4, 5, 6, 8\},~
		\{ 0, 1, 2, 3, 4, 7 \}\}$.\\
\bottomrule 
\end{tabular} 
\label{table:list271}
\end{table}

%%% Table for Family 27_2 in length 27
\begin{table}
\caption{List of $\alpha$ and $S_0,S_1$ for $\Gamma(3,9,\alpha,S_0,S_1)$ corresponding to $\mathcal{F}_{27,2}$ in Proposition~\ref{prop27}.}
\centering
\small
\begin{tabular}{c c l}
		\multicolumn{3}{c}{Family $\mathcal{F}_{27,2}$} \\
		\midrule
		$\alpha$ & $S_0$ & \multicolumn{1}{c}{$\{S_1\}$} \\ 
		\midrule
		$4$ & $\{3,6\}$ & $\{\{1, 3, 5, 7 \},\,
		\{1, 3, 5, 8 \},\,
		\{1, 2, 3, 4 \}$,\\
		&&$~~\{2, 3, 7, 8 \},\,
		\{0, 1, 7, 8 \},\,
		\{0, 1, 2, 8 \}\}$. \\
		
		& & $\{\{0, 2, 5, 7 \},\,
		\{1, 4, 6, 8 \},\,
		\{0, 1, 4, 5 \}$,\\
		&&$~~\{3, 4, 7, 8 \},\,
		\{4, 5, 6, 7 \},\,
		\{0, 4, 5, 8 \}\}$.\\
		
		& & $\{\{0, 2, 4, 7 \},\,
		\{1, 2, 6, 7 \},\,
		\{1, 2, 5, 6 \}$,\\
		&&$~~\{2, 4, 6, 8 \},\,
		\{5, 6, 7, 8 \},\,
		\{2, 3, 4, 5 \}\}$. \\
		
		& $\{1, 2, 3, 6, 7, 8\}$ & $\{\{3, 5 \},\,
		\{4, 6 \},\,
		\{0, 2 \},\,
		\{1, 3 \},\,
		\{0, 7 \},\,
		\{6, 8 \}\}$. \\
		
		& $\{1, 3, 4, 5, 6, 8\}$ & $\{\{6, 7 \},\,
		\{0, 8 \},\,
		\{3, 4 \},\,
		\{2, 3 \},\,
		\{5, 6 \},\,
		\{0, 1 \}\}$. \\
		
		& $\{2, 3, 4, 5, 6, 7\}$ & $\{\{0, 5 \},\,
		\{1, 6 \},\,
		\{2, 6 \},\,
		\{3, 8 \},\,
		\{3, 7 \},\,
		\{0, 4 \}\}$. \\
		
		& $\{3, 4, 5, 6\}$ & $\{\{1, 3, 4 \},\,
		\{0, 1, 7 \},\,
		\{5, 6, 8 \}$,\\
		&&$~~\{0, 2, 8 \},\,
		\{4, 6, 7 \},\,
		\{2, 3, 5 \}\}$. \\
		
		& $\{2, 3, 6, 7\}$ & $\{\{2, 3, 8 \},\,
		\{1, 6, 7 \},\,
		\{3, 4, 7 \}$,\\
		&&$~~\{2, 5, 6 \},\,
		\{0, 1, 4 \},\,
		\{0, 5, 8 \}\}$. \\
		
		& $\{1, 3, 6, 8\}$ & $\{\{2, 6, 8 \},\,
		\{0, 4, 7 \},\,
		\{1, 3, 7 \}$,\\
		&&$~~\{1, 4, 6 \},\,
		\{0, 2, 5 \},\,
		\{3, 5, 8 \}\}$. \\
		
		\midrule
		
		$7$ & $\{3,6\}$ & $\{\{1, 3, 5, 7 \},\,
		\{1, 3, 5, 8 \},\,
		\{1, 2, 3, 4 \}$,\\
		&&$~~\{2, 3, 7, 8 \},\,
		\{0, 1, 7, 8 \},\,
		\{0, 1, 2, 8 \}\}$. \\
		
		& & $\{\{0, 2, 5, 7 \},\,
		\{1, 4, 6, 8 \},\,
		\{0, 1, 4, 5 \}$,\\
		&&$~~\{3, 4, 7, 8 \},\,
		\{4, 5, 6, 7 \},\,
		\{0, 4, 5, 8 \}\}$. \\
		
		& & $\{\{0, 2, 4, 7 \},\,
		\{1, 2, 6, 7 \},\,
		\{1, 2, 5, 6 \}$,\\
		&&$~~\{2, 4, 6, 8 \},\,
		\{5, 6, 7, 8 \},\,
		\{2, 3, 4, 5 \}\}$. \\
		
		& $\{1, 2, 3, 6, 7, 8\}$ & $\{\{0, 5 \},\,
		\{1, 6 \},\,
		\{2, 6 \},\,
		\{3, 8 \},\,
		\{3, 7 \},\,
		\{0, 4 \}\}$ \\
		
		& $\{1, 3, 4, 5, 6, 8\}$ & $\{\{3, 5 \},\,
		\{4, 6 \},\,
		\{0, 2 \},\,
		\{1, 3 \},\,
		\{0, 7 \},\,
		\{6, 8 \}\}$ \\
		
		& $\{2, 3, 4, 5, 6, 7\}$ & $\{\{6, 7 \},\,
		\{0, 8 \},\,
		\{3, 4 \},\,
		\{2, 3 \},\,
		\{5, 6 \},\,
		\{0, 1 \}\}$ \\
		
		& $\{3, 4, 5, 6\}$ & $\{\{2, 6, 8 \},\,
		\{0, 4, 7 \},\,
		\{1, 3, 7 \}$,\\
		&&$~~\{1, 4, 6 \},\,
		\{0, 2, 5 \},\,
		\{3, 5, 8 \}\}$. \\
		
		& $\{2, 3, 6, 7\}$ & $\{\{1, 3, 4 \},\,
		\{0, 1, 7 \},\,
		\{5, 6, 8 \}$,\\
		&&$~~\{0, 2, 8 \},\,
		\{4, 6, 7 \},\,
		\{2, 3, 5 \}\}$. \\
		
		& $\{1, 3, 6, 8\}$ & $\{\{2, 3, 8 \},\,
		\{1, 6, 7 \},\,
		\{3, 4, 7 \}$,\\
		&&$~~\{2, 5, 6 \},\,
		\{0, 1, 4 \},\,
		\{0, 5, 8 \}\}$. \\
\bottomrule 
\end{tabular} 
\label{table:list272}
\end{table}

%%%

\begin{table}
\caption{The Weight distribution of $C_{27,1}$ and $C_{27,2}$}
\centering
\small
\begin{tabular}{c r |c r | c r}
\hline
\multicolumn{6}{c}{$C_{27,1}$}\\
\hline
		${\rm wt}$ & $\# \mathbf{c}$ & 
		${\rm wt}$ & $\# \mathbf{c}$ &
		${\rm wt}$ & $\# \mathbf{c}$\\ 
		\hline
		$0$ & $1$ & $15$ & $1,857,060$ & $22$ & $18,873,108$ \\
		$9$ & $591$ & $16$ & $4,188,213$ & $23$ & $12,305,844$ \\
		$10$ & $4,077$ & $17$ & $8,123,490$ & $24$ & $6,155,010$ \\
		$11$ & $17,901$ & $18$ & $13,519,830$ & $25$ & $2,217,159$ \\
		$12$ & $68,868$ & $19$ & $19,215,414$ & $26$ & $511,461$ \\
		$13$ & $237,276$ & $20$ & $23,076,036$ & $27$ & $56,581$ \\
		$14$ & $712,260$ & $21$ & $23,077,548$ & & \\
		\hline 
		\multicolumn{6}{c}{$C_{27,2}$}\\
		\hline
		${\rm wt}$ & $\# \mathbf{c}$ & 
		${\rm wt}$ & $\# \mathbf{c}$ &
		${\rm wt}$ & $\# \mathbf{c}$\\ 
		\hline
		$0$ & $1$ & $15$ & $1,846,476$ & $22$ & $18,873,108$ \\
		$9$ & $717$ & $16$ & $4,188,213$ & $23$ & $12,301,308$ \\
		$10$ & $4,077$ & $17$ & $8,139,366$ & $24$ & $6,155,010$ \\
		$11$ & $16,767$ & $18$ & $13,519,830$ & $25$ & $2,218,293$ \\
		$12$ & $68,868$ & $19$ & $19,199,538$ & $26$ & $511,461$ \\
		$13$ & $241,812$ & $20$ & $23,076,036$ & $27$ & $56,455$ \\
		$14$ & $712,260$ & $21$ & $23,088,132$ & & \\
		\hline 
\end{tabular} 
\label{table:27}
\end{table}

A self-dual additive $\left(36,2^{36},12\right)_4$ code from circulant graphs does not exist, since the best minimum distance is confirmed to be $11$ in \cite{Grassl2017,Saito2019}. The metacirculant graph construction increases the known minimum distance of length $36$ additive self-dual codes to $12$.

\begin{proposition}\label{prop36}
There are $72$ metacirculant graphs with $(m,n)=(2,18)$, producing two inequivalent additive symplectic self-dual $\left(36, 2^{36}, 12\right)_4$ Type ${\rm II}$ codes. We can separate the $72$ graphs into two families, each containing $36$ isomorphic graphs. They yield $\dsb{36,0,12}_2$ qubit codes. The respective families $\mathcal{F}_{36,1}$ and $\mathcal{F}_{36,2}$ are represented by the graphs 
\begin{align}
G_{36,1} &:=\Gamma(2, 18, 1, \{ 4, 6, 12, 14\}, \{ 3, 4, 6, 7, 9, 11, 12, 14, 15 \}), \notag \\
G_{36,2} &:=\Gamma(2, 18, 1, \{ 4, 6, 12, 14\}, \{ 1, 4, 7, 8, 9, 10, 11, 14, 17 \}).
\end{align}
\end{proposition}

\begin{proof}
The weight distributions of the self-dual codes $C_{36,1}$ and $C_{36,2}$ derived, respectively, from $G_{36,1}$ and $G_{36,2}$ are in Table~\ref{table:36}. The list of $\{\alpha, S_0, S_1\}$ for the metacirculant graphs 
$\Gamma(2, 18,  \alpha, S_0, S_1)$ that splits into two families is given in Table~\ref{table:list36}. The two graphs $G_{36,1}$ and $G_{36,2}$ are presented in Table~\ref{graph:36}. They share common structures within each of the two layers and differ only on the edges connecting distinct layers.
\end{proof}

\begin{table}
\caption{The Weight distribution of $C_{36,1}$ and $C_{36,2}$}
\centering
\small
\begin{tabular}{c r |c r | c r}
		\hline
		\multicolumn{6}{c}{$C_{36,1}$}\\
		\hline 
		${\rm wt}$ & $\# \mathbf{c} \in C$ & 
		${\rm wt}$ & $\# \mathbf{c} \in C$ &
		${\rm wt}$ & $\# \mathbf{c} \in C$\\ 
		\hline
		$0$ & $1$ & $20$ & $746,262,396$ & $30$ & $11,670,433,632$ \\
		$12$ & $28,764$ & $22$ & $3,459,817,152$ & $32$ & $3,176,936,829$ \\
		$14$ & $425,952$ & $24$ & $10,296,739,656$ & $34$ & $305,789,472$ \\
		$16$ & $9,744,570$ & $26$ & $18,797,790,528$ & $36$ & $4,362,804$ \\
		$18$ & $100,283,040$ & $28$ & $20,150,861,940$ & & \\
		\hline 
		\multicolumn{6}{c}{$C_{36,2}$}\\
		\hline
		${\rm wt}$ & $\# \mathbf{c} \in C$ & 
		${\rm wt}$ & $\# \mathbf{c} \in C$ &
		${\rm wt}$ & $\# \mathbf{c} \in C$\\ 
		\hline
		$0$ & $1$ & $20$ & $742,341,996$ & $30$ & $11,672,176,032$ \\
		$12$ & $208,44$ & $22$ & $3,466,089,792$ & $32$ & $3,176,414,109$ \\
		$14$ & $520,992$ & $24$ & $10,289,421,576$ & $34$ & $305,884,512$ \\
		$16$ & $9,221,850$ & $26$ & $18,804,063,168$ & $36$ & $4,354,884$ \\
		$18$ & $102,025,440$ & $28$ & $20,146,941,540$ & & \\
		\hline 
\end{tabular} 
\label{table:36}
\end{table}

%%% List of two families of length 36
\begin{table}
	\caption{List of input parameters $\alpha$ and $S_0,S_1$ for the metacirculant graphs $\Gamma(2,18,\alpha,S_0,S_1)$ in Proposition~\ref{prop36}.}
	\renewcommand{\arraystretch}{1.05}
	\setlength{\tabcolsep}{3pt}
	\centering
	\footnotesize
	\begin{tabular}{c l | c l }
		\toprule
		\multicolumn{4}{c}{Family $\mathcal{F}_{36,1}$}  \\
		\hline
		$\alpha$ & \multicolumn{1}{c|}{$S_0,S_1$} & 
		$\alpha$ & \multicolumn{1}{c}{$S_0,S_1$} \\ 
		\midrule
		
		$1$ & $ \{ 4, 6, 12, 14 \},\{ 3, 4, 6, 7, 9, 11, 12, 14, 15 \}$ &
		$17$ & $\{ 2, 6, 12, 16 \}, \{ 1, 4, 7, 8, 9, 11, 12, 13, 16 \}$\\ 
		
		& $ \{2,6,12,16\},\{ 1, 2, 3, 6, 9, 12, 15, 16, 17 \}$ & 
		& $\{ 2, 6, 12, 16 \}, \{ 0, 1, 4, 7, 10, 13, 14, 15, 17 \}$\\		
		
		& $ \{ 6, 8, 10, 12 \}, \{ 0, 1, 3, 4, 6, 12, 14, 15, 17 \}$ &
		& $\{ 2, 6, 12, 16 \}, \{ 2, 5, 6, 7, 9, 10, 11, 14, 17 \}$ \\
		
		& $ \{ 2, 3, 6, 7, 11, 12, 15, 16 \},\{ 0, 6, 8, 10, 12 \}$ &
		& $\{ 2, 6, 12, 16 \}, \{ 2, 5, 8, 11, 12, 13, 15, 16, 17 \}$ \\
		
		& $\{ 1, 3, 6, 8, 10, 12, 15, 17 \},\{ 3, 5, 9, 13, 15 \}$ &
		& $\{ 2, 6, 12, 16 \}, \{ 1, 2, 3, 5, 6, 7, 10, 13, 16 \}$ \\
		
		& $\{ 3, 4, 5, 6, 12, 13, 14, 15 \},\{ 3, 7, 9, 11, 15 \}$ &
		& $\{ 2, 6, 12, 16 \}, \{ 0, 1, 3, 4, 5, 8, 11, 14, 17 \}$ \\
		
		& $\{ 3, 5, 6, 8, 9, 10, 12, 13, 15 \},\{ 1, 3, 15, 17 \}$ & 
		& $\{ 2, 6, 12, 16 \}, \{ 0, 3, 6, 9, 12, 13, 14, 16, 17 \}$ \\
		
		& $\{ 3, 4, 6, 7, 9, 11, 12, 14, 15 \},\{ 3, 5, 13, 15 \}$ &
		& $\{ 2, 6, 12, 16 \}, \{ 0, 1, 2, 4, 5, 6, 9, 12, 15 \}$ \\
		
		& $\{ 1, 2, 3, 6, 9, 12, 15, 16, 17 \},\{ 2, 6, 12, 16 \}$ & 
		& $\{ 2, 6, 12, 16 \}, \{ 0, 3, 6, 7, 8, 10, 11, 12, 15 \}$ \\
		
		$17$ & $\{ 4, 6, 12, 14 \},\{ 0, 1, 3, 5, 6, 8, 9, 15, 16 \}$ & 
		& $\{ 6, 8, 10, 12 \}, \{ 0, 2, 8, 10, 11, 13, 14, 15, 17 \}$ \\
		
		& $\{ 4, 6, 12, 14 \},\{ 5, 6, 8, 9, 11, 13, 14, 16, 17 \}$ &
		& $\{ 6, 8, 10, 12 \}, \{ 3, 5, 6, 8, 9, 10, 12, 13, 15 \}$ \\
		
		& $\{ 4, 6, 12, 14 \},\{ 0, 1, 3, 4, 6, 8, 9, 11, 12 \}$ &
		& $\{ 6, 8, 10, 12 \}, \{ 0, 2, 3, 5, 11, 13, 14, 16, 17 \}$ \\ 
		
		& $\{ 4, 6, 12, 14 \},\{ 1, 2, 4, 6, 7, 9, 10, 16, 17 \}$ &
		& $\{ 6, 8, 10, 12 \}, \{ 0, 2, 3, 4, 6, 7, 9, 15, 17 \}$ \\ 		
		
		& $\{ 4, 6, 12, 14 \},\{ 1, 3, 4, 6, 7, 13, 14, 16, 17 \}$ &
		& $\{ 6, 8, 10, 12 \}, \{ 1, 7, 9, 10, 12, 13, 14, 16, 17 \}$ \\		
		
		& $\{ 4, 6, 12, 14 \}, \{ 2, 3, 5, 6, 8, 10, 11, 13, 14 \}$ &
		& $\{ 6, 8, 10, 12 \}, \{ 0, 2, 3, 5, 6, 7, 9, 10, 12 \}$ \\		
		
		& $\{ 4, 6, 12, 14 \}, \{ 3, 4, 6, 7, 9, 11, 12, 14, 15 \}$ &
		& $\{ 6, 8, 10, 12 \}, \{ 1, 2, 4, 10, 12, 13, 15, 16, 17 \}$ \\		
		
		& $\{ 4, 6, 12, 14 \}, \{ 0, 1, 3, 4, 10, 11, 13, 14, 16 \}$ &
		& $\{ 6, 8, 10, 12 \}, \{ 0, 1, 2, 4, 5, 7, 13, 15, 16 \}$ \\		
		
		& $\{ 4, 6, 12, 14 \}, \{ 0, 2, 3, 5, 7, 8, 10, 11, 17 \}$ &
		& $\{ 6, 8, 10, 12 \}, \{ 1, 2, 3, 5, 6, 8, 14, 16, 17 \}$ \\ 		
\toprule 
\multicolumn{4}{c}{Family $\mathcal{F}_{36,2}$}  \\
\hline
$\alpha$ & \multicolumn{1}{c}{$S_0,S_1$} & 
$\alpha$ & \multicolumn{1}{c}{$S_0,S_1$} \\ 
\midrule
$1$ & $\{ 4, 6, 12, 14 \}, \{ 1, 4, 7, 8, 9, 10, 11, 14, 17 \}$ &
$17$ & $\{ 2, 6, 12, 16 \}, \{ 1, 2, 4, 5, 9, 13, 14, 16, 17 \}$ \\

& $ \{ 2, 6, 12, 16 \},\{ 1, 2, 4, 5, 9, 13, 14, 16, 17 \}$ &
& $\{ 2, 6, 12, 16 \}, \{ 1, 5, 6, 8, 9, 11, 12, 14, 15 \}$ \\

& $ \{ 6, 8, 10, 12 \}, \{ 0, 1, 2, 4, 7, 11, 14, 16, 17 \}$ & 
& $\{ 2, 6, 12, 16 \}, \{ 3, 7, 8, 10, 11, 13, 14, 16, 17 \}$ \\

& $\{ 2, 5, 6, 7, 11, 12, 13, 16 \},\{ 0, 4, 8, 10, 14 \}$ &
& $\{ 2, 6, 12, 16 \}, \{ 3, 4, 6, 7, 9, 10, 12, 13, 17 \}$ \\

& $\{ 1, 4, 5, 6, 12, 13, 14, 17 \},\{ 1, 7, 9, 11, 17 \}$ &
& $\{ 2, 6, 12, 16 \}, \{ 0, 2, 3, 5, 6, 8, 9, 13, 17 \}$ \\

& $\{ 1, 6, 7, 8, 10, 11, 12, 17 \},\{ 5, 7, 9, 11, 13 \}$ &
& $\{ 2, 6, 12, 16 \}, \{ 0, 1, 3, 4, 6, 7, 11, 15, 16 \}$ \\

& $\{ 1, 4, 6, 7, 9, 11, 12, 14, 17 \},\{ 4, 8, 10, 14 \}$ &
& $\{ 2, 6, 12, 16 \}, \{ 1, 2, 6, 10, 11, 13, 14, 16, 17 \}$ \\

& $\{ 1, 2, 5, 6, 9, 12, 13, 16, 17 \},\{ 5, 7, 11, 13 \}$ &
& $\{ 2, 6, 12, 16 \}, \{ 0, 1, 3, 4, 6, 7, 9, 10, 14 \}$ \\

& $\{ 5, 6, 7, 8, 9, 10, 11, 12, 13 \},\{ 1, 7, 11, 17 \}$ &
& $\{ 2, 6, 12, 16 \}, \{ 2, 3, 5, 6, 8, 9, 11, 12, 16 \}$ \\

$17$ & $\{ 4, 6, 12, 14 \},\{ 1, 3, 6, 9, 10, 11, 12,13,16\}$ &
& $\{ 6, 8, 10, 12 \}, \{ 2, 6, 9, 11, 12, 13, 14, 15, 17 \}$ \\

& $\{ 4, 6, 12, 14 \},\{ 1, 4, 6, 9, 12, 13, 14, 15, 16 \}$ &
& $\{ 6, 8, 10, 12 \}, \{ 0, 3, 7, 10, 12, 13, 14, 15, 16 \}$ \\

& $\{ 4, 6, 12, 14 \},\{ 2, 5, 6, 7, 8, 9, 12, 15, 17 \}$ &
& $\{ 6, 8, 10, 12 \}, \{ 0, 4, 7, 9, 10, 11, 12, 13, 15 \}$ \\ 

& $\{ 4, 6, 12, 14 \},\{ 1, 4, 5, 6, 7, 8, 11, 14, 16 \}$ &
& $\{ 6, 8, 10, 12 \}, \{ 1, 2, 3, 4, 5, 7, 10, 14, 17 \}$ \\

& $\{ 4, 6, 12, 14 \},\{ 1, 2, 3, 4, 5, 8, 11, 13, 16 \}$ &
& $\{ 6, 8, 10, 12 \}, \{ 2, 5, 7, 8, 9, 10, 11, 13, 16 \}$ \\ 

& $\{ 4, 6, 12, 14 \}, \{ 1, 4, 7, 8, 9, 10, 11, 14, 17 \}$ & 
& $\{ 6, 8, 10, 12 \}, \{ 3, 6, 8, 9, 10, 11, 12, 14, 17 \}$\\ 

& $\{ 4, 6, 12, 14 \}, \{ 0, 1, 2, 3, 6, 9, 11, 14, 17 \}$ & 
& $\{ 6, 8, 10, 12 \}, \{ 0, 1, 3, 6, 10, 13, 15, 16, 17 \}$ \\ 

& $\{ 4, 6, 12, 14 \}, \{ 2, 3, 4, 5, 6, 9, 12, 14, 17 \}$ &
& $\{ 6, 8, 10, 12 \}, \{ 0, 2, 5, 9, 12, 14, 15, 16, 17 \}$ \\

& $\{ 4, 6, 12, 14 \}, \{ 0, 3, 6, 7, 8, 9, 10, 13, 16 \}$ & 
& $\{ 6, 8, 10, 12 \}, \{ 2, 4, 5, 6, 7, 8, 10, 13, 17 \}$ \\
\bottomrule
\end{tabular} 
\label{table:list36}
\end{table}

%%%% Graphs for the two families
\begin{table}
	\caption{List of Edges of $G_{36,1}$ and $G_{36,2}$ in Proposition~\ref{prop36}. 
		Vertices $1$ to $18$ are in $\mathcal{L}_1$. Vertices $19$ to $36$ are in $\mathcal{L}_2$.}
	\centering
	\begin{tabular}{c}
		\toprule
		The $36$ edges $\{(i,j)\}$ as $(i,\{j \in J\})$ with $i \in \mathcal{L}_1$ and $J \subset \mathcal{L}_1$.\\
		Common to both $G_{36,1}$ and $G_{36,2}$.\\
		$(1,\{5,7,13,15\}), \,(2,\{6,8,14,16\}), \,(3,\{7,9,15,17\})$,\\
		$(4,\{8,10,16,18\}), \,(5,\{9,11,17\}), \,(6,\{10,12,18\})$,\\
		$(7,\{11,13\}), \,(8,\{12,14\}), \,(9,\{13,15\}), \,(10,\{14,16\})$,\\
		$(11,\{15,17\}), \,(12,\{16,18\}), \,(13,\{17\}), \,(14,\{18\})$.\\
		
		\midrule
		
		The $36$ edges $\{(i,j)\}$ as $(i,\{J\})$ with $i \in \mathcal{L}_2$ and $j \in J \subset \mathcal{L}_2$.\\
		Common to both $G_{36,1}$ and $G_{36,2}$.\\
		$(19,\{23,25,31,33\}), \,(20,\{24,26,32,34\}), \,(21,\{ 25,27,33,35\})$,\\ $(22,\{26,28,34,36\}), \,(23,\{27,29,35\}), \,(24,\{28,30,36\})$,\\
		$(25,\{29,31\}), \,(26,\{30,32\}), \,(27,\{31,33\}), \,(28,\{32,34\})$,\\ $(29,\{33,35\}), \,(30,\{34,36\}), \, (31,\{35\}), \,(32,\{36\})$.\\
		
		\midrule
		
		The $162$ edges $\{(i,j)\}$ as $(i,\{J\})$ with $i \in \mathcal{L}_1$ and $j \in J \subset \mathcal{L}_2$\\
		\midrule
		In $G_{36,1}$ \\
		\midrule
		$(1, \{22,23,25,26,28,30,31,33,34\}), \,
		(2, \{23,24,26,27,29,31,32,34,35)\},$ \\ 
		
		$(3,\{24,25,27,28,30,32,33,35,36\}), \,
		(4,\{19,25,26,28,29,31,33,34,36\}),$ \\
		
		$(5,\{19,20,26,27,29,30,32,34,35\}), \,
		(6,\{20,21,27,28,30,31,33,35,36\}),$ \\
		
		$(7,\{19,21,22,28,29,31,32,34,36\}), \,
		(8, \{19,20,22,23,29,30,32,33,35\}),$ \\
		
		$(9, \{20,21,23,24,30,31,33,34,36\}), \, 
		(10, \{19,21,22,24,25,31,32,34,35\}),$ \\
		
		$(11,\{20,22,23,25,26,32,33,35,36\}), \,
		(12,\{19,21,23,24,26,27,33,34,36\}),$ \\
		
		$(13,\{19,20,22,24,25,27,28,34,35\}), \,
		(14,\{20,21,23,25,26,28,29,35,36\}),$ \\
		
		$(15,\{19,21,22,24,26,27,29,30,36\}), \,
		(16,\{19,20,22,23,25,27,28,30,31\}),$ \\
		
		$(17,\{20,21,23,24,26,28,29,31,32\}), \,
		(18,\{21,22,24,25,27,29,30,32,33\})$ \\
		
		\midrule
		
		In $G_{36,2}$\\		
	
	\midrule
	
	$(1, \{20,23,26,27,28,29,30,33,36\}), \,
	(2, \{19,21,24,27,28,29,30,31,34\}),$\\ 
	
	$(3, \{20,22,25,28,29,30,31,32,35\}), \,
	(4, \{21,23,26,29,30,31,32,33,36\}),$ \\
	
	$(5, \{19,22,24,27,30,31,32,33,34\}), \,
	(6, \{20,23,25,28,31,32,33,34,35\}),$ \\
	
	$(7, \{21,24,26,29,32,33,34,35,36\}), \,
	(8, \{19,22,25,27,30,33,34,35,36\}),$ \\
	
	$(9, \{19,20,23,26,28,31,34,35,36\}), \,
	(10, \{19,20,21,24,27,29,32,35,36\}),$ \\
	
	$(11, \{19,20,21,22,25,28,30,33,36\}), \,
	(12, \{19,20,21,22,23,26,29,31,34\}),$ \\
	
	$(13, \{20,21,22,23,24,27,30,32,35\}), \,
	(14, \{21,22,23,24,25,28,31,33,36\}),$ \\
	
	$(15, \{19,22,23,24,25,26,29,32,34\}), \,
	(16, \{20,23,24,25,26,27,30,33,35\}),$ \\
	
	$(17, \{21,24,25,26,27,28,31,34,36\}), \,
	(18, \{19,22,25,26,27,28,29,32,35\}),$ \\
	
	\bottomrule 
\end{tabular} 
\label{graph:36}
\end{table}

%%%%%%%%%%%%%%%%%%%%%%%%%%% New Section %%%%%%%%%%%%%%%%%%%%%%%%%%%%%
\section{Improved qubit codes}\label{sec:betterqubit}

Here we present improved qubit codes $\dsb{\ell,0,d}_2$ for lengths $\ell \in \{78,90,91,93,96\}$. Due to the large sizes of the relevant graphs, their presentation will be given in the appendix. 

Applying secondary constructions yields more qubit codes with strictly better parameters that previously known. There are computational routines used by M. Grassl to perform the propagation rules on the improved codes submitted for inclusion to his online table \cite{Grassl:codetables}. We highlight the process only for $\ell=78$, for brevity.

\begin{proposition}\label{prop78}
The metacirculant graph 
\begin{multline}
G_{78}:=\Gamma(6, 13, 12, \{1,4,6,7,9,12\}, 
\{1,2,3,5,7,8,9,11\},\\
\{1,2,4,5,8,10\}, \{1,2,3,6,7,8,12\})
\end{multline}
yields a new $\left(78, 2^{78}, 20\right)_4$ Type II additive self-dual code. The corresponding $Q_{78}$ is a $\dsb{78,0,20}_2$ qubit code. Prior to our discovery, the best-known was $\dsb{78,0,19}_2$ and the first known occurence of $d=20$ was at $\ell = 80$. Our improved qubit code has been listed as the current best-known since October 13, 2020 in Grassl's Table~\cite{Grassl:codetables} based on a private communication.
\end{proposition}

\begin{proof}
A randomized search in {\tt MAGMA} found $G_{78}$. The graph yields the self-dual additive code $C_{78}$ of distance $d = 20$. By Theorem~\ref{thm:even}, $\Delta_S = |S_0| + |S_3| = 13$, which implies that $C_{78}$ is Type II. The qubit code $Q_{78}$ can then be certified to have the claimed parameters. For an easier inspection, since $G_{78}$ is too complex to visualize, we list the edges in Table~\ref{graph:78}. 
\end{proof}

By propagation rules, strict improvements were also achieved by puncturing, shortening, lengthening, and taking a subcode of $Q_{78}$. Table~\ref{table:length78} provides the details.

\begin{proposition}\label{prop90}
The metacirculant graph 
\begin{multline}
G_{90}:=\Gamma(10, 9, 8, \{ 1, 8 \},\{ 0, 1, 2, 4, 5, 8 \}, 
\{ 5, 6 \},\\ 
\{ 2, 4, 5, 6, 8 \},\{ 0, 1, 2, 4, 7 \},\{ 0, 5, 7, 8 \})
\end{multline}
generates a new $\left(90,2^{90},21\right)_4$ additive self-dual code $C_{90}$. The resulting $\dsb{90,0,21}$ qubit code $Q_{90}$ has better minimum distance than the best-known $\dsb{90,0,20}$ in \cite{Grassl:codetables}.
\end{proposition}

\begin{proof}
The graph $G_{90}$ and the resulting codes $C_{90}$ and $Q_{90}$ as well as their parameters were found by {\tt MAGMA} searches. The code $C_{90}$ is Type I by Theorem~\ref{thm:even}, since $\Delta_S = |S_0|+|S_5|=6$. For verification, $G_{90}$ is described by its list of edges in Table~\ref{graph:90}.
\end{proof}

\begin{proposition}\label{prop91}
The multi-partite metacirculant graph
\begin{multline}	
G_{91}:= \Gamma(7, 13, 3, \{ \},\{ 4, 7, 8, 10, 11, 12 \},\\
		\{ 1, 3, 4, 7, 8, 9, 10, 11, 12 \}, \{ 0, 4, 7, 8, 10, 11, 12 \})	
\end{multline}
generates a new $\left(91, 2^{91}, 22\right)_4$ additive self-dual Type I code $C_{91}$. The corresponding $\dsb{91, 0, 22}_2$ qubit code $Q_{91}$ improves the minimum distance of the $\dsb{91,0,21}$ code in \cite{Grassl:codetables}.
\end{proposition}

\begin{proof}
By Theorem~\ref{thm:mult}, $G_{91}$ is a multi-partite graph. The vertices are partitioned into $7$ layers, each containing $13$ vertices. We used {\tt MAGMA} to verify that the minimum distance $d$ of the generated $[91,45.5,d]_4$ self-dual additive code $C_{91}$ is indeed $22$. This code is clearly Type I since its length is odd. The corresponding $\dsb{91,0,22}_2$ qubit code $Q_{91}$ improves on the $\dsb{91,0,21}$ code currently listed as best-known in \cite{Grassl:codetables}. The list of the neighbours of each vertex in $G_{91}$ is in Table~\ref{table:List91}.
\end{proof}

%%%%%%%%%%%%
\begin{proposition}\label{prop93}
The metacirculant graph
\begin{multline}	
	G_{93}:= \Gamma(3, 31, 1,  \{ 10, 12, 14, 15, 16, 17, 19, 21 \}, \\
	\{ 2, 7, 8, 10, 11, 13, 15, 16, 17, 19 \})	
\end{multline}
generates a new $\left(93, 2^{93}, 22\right)_4$ additive self-dual code $C_{93}$. The code is Type I due to its odd length. The resulting $\dsb{93, 0, 21}_2$ qubit code $Q_{93}$ has an improved minimum distance compared with the prior best known $\dsb{93,0,20}_2$ code in \cite{Grassl:codetables}.
\end{proposition}

\begin{proof}
We used {\tt MAGMA} to verify that the generated $\left(93,2^{93},d\right)_4$ additive self-dual code $C_{93}$ has minimum distance $d = 21$. This Type I code corresponds to a $\dsb{93, 0, 21}_2$ qubit code $Q_{93}$. For verification, Table~\ref{list:93} lists the edges of $G_{93}$.
\end{proof}

%%%%%%%%%%%%

\begin{proposition}\label{prop96}
The metacirculant graph
\begin{multline}
G_{96}:= \Gamma(6, 16, 7, 
\{ 2, 4, 6, 7, 9, 10, 12, 14 \},  
\{ 1, 3, 4, 5, 7, 9, 10 \},\\
\{ 3, 5, 11, 14 \},\{ 0, 2, 7, 10, 15\})	
\end{multline}
generates a new $\left(96, 2^{96}, 22\right)_4$ additive self-dual Type II code $C_{96}$. The resulting $\dsb{96, 0, 22}_2$ qubit code $Q_{96}$ improves the minimum distance of the prior best-known $\dsb{96,0,20}_2$ quantum code in \cite{Grassl:codetables}. 
\end{proposition}

\begin{proof}
Table~\ref{table:List96} lists the neighbours of each vertex of $G_{96}$, for verification. We confirmed that the minimum distance of $C_{96}$ is $22$ by using {\tt MAGMA}. Since $\Delta_S = |S_0|+ |S_3| = 13$, the code is Type II by Theorem~\ref{thm:even}.
\end{proof}

\section{Concluding Remarks}
There are at least two challenging aspects in finding improved symplectic self-dual additive codes over $\mathbb{F}_4$ or strictly better qubit codes. First, the search space quickly widens as the length $\ell$ grows. Second, determining the minimum distances of the codes is computationally expensive. 

We have shown that focusing on metacirculant graphs which are not isomorphic to the circulant ones is a fruitful approach. 

Bringing this work to its end, Table~\ref{table:property} summarizes the properties of the graphs that we have used above. All of them have diameter $2$ and girth $3$. Included in the table, for each graph $G$, are the minimum distance $d_{\rm min}(G)$ of the generated additive code $C:=C(G)$, the valency $\nu(G)$, the size $\gamma(G)$ of the maximum clique, and the size $|{\rm Aut}(G)|$ of the automorphism group of $G$.

\begin{table}
\caption{Properties of the Graphs}
\centering
\small
\begin{tabular}{c c c c c}
		\toprule
		Graph $G$ & $d_{\rm min}(G)$ & $\nu(G)$ & $\gamma(G)$ & $|{\rm Aut}(G)|$ \\
		\midrule
		$G_{12}$ & $6$ & $5$ & $4$ & $24$ \\
		
		$G_{27,1}$ & $9$ & $16$ & $6$ & $27$ \\
		$G_{27,2}$ & $9$ & $10$ & $4$ & $27$ \\
		$G_{36,1}$ & $12$ & $13$ & $6$ & $72$ \\
		$G_{36,2}$ & $12$ & $13$ & $4$ & $72$ \\
		
		$G_{78}$ & $20$ & $41$ & $7$ & $78$ \\
		$G_{90}$ & $21$ & $42$ & $7$ & $90$ \\
		$G_{91}$ & $22$ & $44$ & $7$ & $546$ \\
		$G_{93}$ & $22$ & $28$ & $4$ & $186$ \\
		$G_{96}$ & $22$ & $35$ & $6$ & $96$ \\
		\bottomrule 
	\end{tabular}
	\label{table:property} 
\end{table}

%%%%%%%%%%
\begin{table}[h!]
\caption{New codes from modifying $Q_{78}$}
\renewcommand{\arraystretch}{1.1}
\centering
\small
\begin{tabular}{c c l} 
\toprule
Code & Parameters & Propagation rule \\% [0.3ex] 
\midrule
		$Q_{78, 1}$ & $\dsb{77,0,19}_2$ & Puncture $Q_{78}$ at $\{78\}$\\
		$Q_{78,2}$ & $\dsb{77,1,19}_2$ & Shorten $Q_{78}$ at $\{78\}$\\
		
		$Q_{78,3}$ & $\dsb{78,1,19}_2$ & Lengthen $Q_{78,2}$ by $1$\\
		
		$Q_{78,4}$ & $\dsb{76,2,18}_2$ & Shorten $Q_{78}$ at $\{77, 78\}$     \\
		
		$Q_{78,5}$ & $\dsb{76,1,18}_2$ & Subcode of $Q_{78,4}$ \\
		
		$Q_{78,6}$ & $\dsb{77,2,18}_2$ & Lengthen $Q_{78,4}$ by $1$ \\
		
		$Q_{78,7}$ & $\dsb{75,3,17}_2$ & Shorten $Q_{78}$ at  $\{76,77,78\}$   \\
		
		$Q_{78,8}$ & $\dsb{76,3,17}_2$ & Lengthen $Q_{78,7}$ by $1$ \\
		
		$Q_{78,9}$ & $\dsb{75, 2, 17}_2$ & Subcode of $Q_{78,7}$ \\
\bottomrule
\end{tabular}
\label{table:length78}

\end{table}
%%%%%%%%%%%%%%%
%For acknowledgements section, please don't number the section, please begin it with \section*{Acknowledgements}

% You may incorporate your references as follows in your main tex file.
% Using BibTex is not recommended but can be handled.

%%%%% Start of Appendix
\section*{Appendix}

The metacirculant graphs defined in Section~\ref{sec:betterqubit} are described here in details for reconstruction and verification, should the need arise. For a cleaner presentation, we relabel the vertices as positive integers instead of as tuples $(i,j)$ with $0 \leq i < m$ and $0 \leq j < n$. Given $m$ and $n$, vertices in partition $V_0$, now called layer $\mathcal{L}_1$, are labeled $1$ to $n$. Vertices $n+1$ to $2n$ are in partition $V_1$, now called layer $\mathcal{L}_2$. We repeat the assignment until all vertices are accounted for. 

For $G_{\ell}$ with $\ell \in \{78,93,96\}$, the edges in $G_{\ell}$ are listed according to their layer(s) to highlight the partitioning of the vertices and how different layers are connected. The numbers of layers for $\ell \in \{90,91\}$ are $> 6$. To keep the description for each graph within a page, we simply list their edges.

\begin{table}
	\caption{The $1\,599$ Edges of $G_{78}$ in Proposition \ref{prop78}: Vertices are partitioned into $6$ layers $\mathcal{L}_1, \ldots,\mathcal{L}_6$, with $13$ vertices in each layer. Vertices $z, z+1, \ldots, z+12$ are in $\mathcal{L}_{k+1}$ for $0 \leq k \leq 5$ and $z=13k+1$.}
	\centering
	\resizebox{\textwidth}{!}{
	\begin{tabular}{cc l}
		\toprule
		$\#$ & $(r,s)$ & $(i,\{j \in J\})$ \\
		\midrule
		$39$ & $(1,1)$ & $(1, \{2, 5, 7, 8, 10, 13\}), 
		(2, \{3, 6, 8, 9,11\}), 
		(3, \{4, 7, 9, 10, 12\}), 
		(4, \{5, 8, 10, 11, 13\}), 
		(5, \{6, 9, 11, 12\}), 
		(6, \{7, 10, 12, 13\}),$\\
		& & $(7, \{8, 11, 13\}),(8, \{9, 12\}), 
		(9, \{10,13\}),(10, \{11\}),(11, \{12\}),(12, \{13\})$\\
		
		$39$ & $(2,2)$ & $(14, \{15, 18, 20, 21, 23, 26\}),(15, \{16, 19, 21, 22, 24\}), (16, \{17, 20, 22, 23, 25\}), (17, \{18, 21, 23, 24, 26\}),(18, \{19, 22, 24, 25\}),$\\
		
		& & $(19, \{20, 23, 25, 26\}),
		(20, \{21, 24, 26\}),(21, \{22, 25\}),
		(22, \{23, 26\}), (23, \{24\}), (24, \{25\}),(25, \{26\})$\\ 
		
		$39$ & $(3,3)$ & $(27, \{28, 31, 33, 34, 36, 39 \}),
		(28, \{29, 32, 34, 35, 37 \}),
		(29, \{30, 33, 35, 36, 38 \}),
		(30, \{31, 34, 36, 37, 39 \}),
		(31, \{32, 35, 37, 38 \}),$\\ 
		&& $(32, \{33, 36, 38, 39 \}),(33, \{34, 37, 39 \}),
		(34, \{35, 38 \}),(35, \{36, 39 \}),(36, \{37 \}),
		(37, \{38 \}),(38, \{39 \})$ \\
		
		$39$ & $(4,4)$ & $(40, \{41, 44, 46, 47, 49, 52 \}), 
		(41, \{42, 45, 47, 48, 50 \}),
		(42, \{43, 46, 48, 49, 51 \}), 
		(43, \{44, 47, 49, 50, 52 \}), 
		(44, \{45, 48, 50, 51 \}),$\\ 
		&& $(45, \{46, 49, 51, 52 \}), 
		(46, \{47, 50, 52 \}),
		(47, \{48, 51 \}),
		(48, \{49, 52 \}), 
		(49, \{50 \}),
		(50, \{51 \}), 
		(51, \{52 \})$\\
		
		$39$ & $(5,5)$ & $(53, \{54, 57, 59, 60, 62, 65 \}), 
		(54, \{55, 58, 60, 61, 63 \}),
		(55, \{56, 59, 61, 62, 64 \}), 
		(56, \{57, 60, 62, 63, 65 \}), 
		(57, \{58, 61, 63, 64 \}),$\\
		
		&& $(58, \{59, 62, 64, 65 \}), 
		(59, \{60, 63, 65 \}), 
		(60, \{61, 64 \}), 
		(61, \{62, 65 \}), 
		(62, \{63 \}), 
		(63, \{64 \}), 
		(64, \{65 \})$\\ 
		
		$39$ & $(6,6)$ & $(66, \{67, 70, 72, 73, 75, 78 \}), 
		(67, \{68, 71, 73, 74, 76 \}),
		(68, \{69, 72, 74, 75, 77 \}), 
		(69, \{70, 73, 75, 76, 78 \}), 
		(70, \{71, 74, 76, 77 \}),$\\
		&& $(71, \{72, 75, 77, 78 \}),
		(72, \{73, 76, 78 \}),
		(73, \{74, 77 \}),
		(74, \{75, 78 \}), 
		(75, \{76 \}),
		(76, \{77 \}),
		(77, \{78 \})$\\
		
		$104$ & $(1,2)$ & 
		$(1, \{15, 16, 17, 19, 21, 22, 23, 25 \}),
		(2, \{16, 17, 18, 20, 22, 23, 24, 26 \}), 
		(3, \{14, 17, 18, 19, 21, 23, 24, 25 \}),
		(4, \{15, 18, 19, 20, 22, 24, 25, 26 \}),$\\
		
		&& $(5, \{14, 16, 19, 20, 21, 23, 25, 26 \}), 
		(6, \{14, 15, 17, 20, 21, 22, 24, 26 \}), 
		(7, \{14, 15, 16, 18, 21, 22, 23, 25 \}), 
		(8, \{15, 16, 17, 19, 22, 23, 24, 26 \}),$\\ 
		
		&& $(9, \{14, 16, 17, 18, 20, 23, 24, 25 \}),
		(10, \{15, 17, 18, 19, 21, 24, 25, 26 \}),
		(11, \{14, 16, 18, 19, 20, 22, 25, 26 \}),$\\ 
		
		&& $(12, \{14, 15, 17, 19, 20, 21, 23, 26 \}),
		(13, \{14, 15, 16, 18, 20, 21, 22, 24 \})$\\ 
		
		$78$ & $(1,3)$ &
		$(1, \{28, 29, 31, 32, 35, 37 \}), 
		(2, \{29, 30, 32, 33, 36, 38 \}), 
		(3, \{30, 31, 33, 34, 37, 39 \}), 
		(4, \{27, 31, 32, 34, 35, 38 \}), 
		(5, \{28, 32, 33, 35, 36, 39 \}),$\\
		
		&& $(6, \{27, 29, 33, 34, 36, 37 \}), 
		(7, \{28, 30, 34, 35, 37, 38 \}), 
		(8, \{29, 31, 35, 36, 38, 39 \}), 
		(9, \{27, 30, 32, 36, 37, 39 \}),$\\
		
		&& $(10, \{27, 28, 31, 33, 37, 38 \}), 
		(11, \{28, 29, 32, 34, 38, 39 \}), 
		(12, \{27, 29, 30, 33, 35, 39 \}), 
		(13, \{27, 28, 30, 31, 34, 36 \}),$\\ 
		
		$91$ & $(1,4)$ & $(1, \{41, 42, 43, 46, 47, 48, 52 \}), 
		(2, \{40, 42, 43, 44, 47, 48, 49 \}), 
		(3, \{41, 43, 44, 45, 48, 49, 50 \}), 
		(4, \{42, 44, 45, 46, 49, 50, 51 \}),$\\
		
		&& $(5, \{43, 45, 46, 47, 50, 51, 52 \}),
		(6, \{40, 44, 46, 47, 48, 51, 52 \}), 
		(7, \{40, 41, 45, 47, 48, 49, 52 \}),
		(8, \{40, 41, 42, 46, 48, 49, 50 \}),$\\
		
		&& $(9, \{41, 42, 43, 47, 49, 50, 51 \}),
		(10, \{42, 43, 44, 48, 50, 51, 52 \}), 
		(11, \{40, 43, 44, 45, 49, 51, 52 \}),$\\
		
		&& $(12, \{40, 41, 44, 45, 46, 50, 52 \}), 
		(13, \{40, 41, 42, 45, 46, 47, 51 \})$\\
		
		$78$ & $(1,5)$ & $(1, \{56, 58, 61, 62, 64, 65 \}), 
		(2, \{53, 57, 59, 62, 63, 65 \}), 
		(3, \{53, 54, 58, 60, 63, 64 \}), 
		(4, \{54, 55, 59, 61, 64, 65 \}), 
		(5, \{53, 55, 56, 60, 62, 65 \}),$\\ 
		
		&& $(6, \{53, 54, 56, 57, 61, 63 \}), 
		(7, \{54, 55, 57, 58, 62, 64 \}), 
		(8, \{55, 56, 58, 59, 63, 65 \}), 
		(9, \{53, 56, 57, 59, 60, 64 \}),$\\ 
		
		&&$(10, \{54, 57, 58, 60, 61, 65 \}), 
		(11, \{53, 55, 58, 59, 61, 62 \}), 
		(12, \{54, 56, 59, 60, 62, 63 \}), 
		(13, \{55, 57, 60, 61, 63, 64 \})$\\ 
		
		$104$ & $(1,6)$ &$(1, \{67, 68, 69, 71, 73, 74, 75, 77 \}), 
		(2, \{68, 69, 70, 72, 74, 75, 76, 78 \}), 
		(3, \{66, 69, 70, 71, 73, 75, 76, 77 \}), 
		(4, \{67, 70, 71, 72, 74, 76, 77, 78 \}),$\\ 
		
		&& $(5, \{66, 68, 71, 72, 73, 75, 77, 78 \}), 
		(6, \{66, 67, 69, 72, 73, 74, 76, 78 \}), 
		(7, \{66, 67, 68, 70, 73, 74, 75, 77 \}), 
		(8, \{67, 68, 69, 71, 74, 75, 76, 78 \}),$ \\
		
		&& $(9, \{66, 68, 69, 70, 72, 75, 76, 77 \}), 
		(10, \{67, 69, 70, 71, 73, 76, 77, 78 \}),
		(11, \{66, 68, 70, 71, 72, 74, 77, 78 \}),$\\ 
		
		&& $(12, \{66, 67, 69, 71, 72, 73, 75, 78 \}), 
		(13, \{66, 67, 68, 70, 72, 73, 74, 76 \})$ \\
		
		$104$ & $(2,3)$ &$(14, \{29, 31, 32, 33, 35, 37, 38, 39 \}), 
		(15, \{27, 30, 32, 33, 34, 36, 38, 39 \}), 
		(16, \{27, 28, 31, 33, 34, 35, 37, 39 \}), 
		(17, \{27, 28, 29, 32\}),$\\ 
		
		&&$(17,\{34, 35, 36, 38 \}),
		(18, \{28, 29, 30, 33, 35, 36, 37, 39 \}), 
		(19, \{27, 29, 30, 31, 34, 36, 37, 38 \}), 
		(20, \{28, 30, 31, 32, 35, 37, 38, 39 \}),$ \\
		
		&&$(21, \{27, 29, 31, 32, 33, 36, 38, 39 \}),
		(22, \{27, 28, 30, 32, 33, 34, 37, 39 \}), 
		(23, \{27, 28, 29, 31, 33, 34, 35, 38 \}),$\\ 
		
		&&$(24, \{28, 29, 30, 32, 34, 35, 36, 39 \}),
		(25, \{27, 29, 30, 31, 33, 35, 36, 37 \}), 
		(26, \{28, 30, 31, 32, 34, 36, 37, 38 \})$\\ 
		
		$78$ & $(2,4)$ & $(14, \{43, 45, 48, 49, 51, 52 \}), 
		(15, \{40, 44, 46, 49, 50, 52 \}), 
		(16, \{40, 41, 45, 47, 50, 51 \}), 
		(17, \{41, 42, 46, 48, 51, 52 \}), 
		(18, \{40, 42, 43\}),$\\
		
		&&$(18, \{ 47, 49, 52 \}),(19, \{40, 41, 43, 44, 48, 50 \}), 
		(20, \{41, 42, 44, 45, 49, 51 \}), 
		(21, \{42, 43, 45, 46, 50, 52 \}), 
		(22, \{40, 43, 44, 46, 47, 51 \}),$ \\
		
		&&$(23, \{41, 44, 45, 47, 48, 52 \}), 
		(24, \{40, 42, 45, 46, 48, 49 \}), 
		(25, \{41, 43, 46, 47, 49, 50 \}), 
		(26, \{42, 44, 47, 48, 50, 51 \})$\\
		
		$91$ & $(2,5)$ & $(14, \{54, 58, 59, 60, 63, 64, 65 \}), 
		(15, \{53, 55, 59, 60, 61, 64, 65 \}), 
		(16, \{53, 54, 56, 60, 61, 62, 65 \}), 
		(17, \{53, 54, 55, 57\}),$\\ 
		
		&&$(17, \{61, 62, 63 \}),
		(18, \{54, 55, 56, 58, 62, 63, 64 \}), 
		(19, \{55, 56, 57, 59, 63, 64, 65 \}), 
		(20, \{53, 56, 57, 58, 60, 64, 65 \}),$\\ 
		
		&&$(21, \{53, 54, 57, 58, 59, 61, 65 \}),
		(22, \{53, 54, 55, 58, 59, 60, 62 \}), 
		(23, \{54, 55, 56, 59, 60, 61, 63 \}),$\\ 
		
		&&$(24, \{55, 56, 57, 60, 61, 62, 64 \}),
		(25, \{56, 57, 58, 61, 62, 63, 65 \}), 
		(26, \{53, 57, 58, 59, 62, 63, 64 \})$\\
		
		$78$ & $(2,6)$ & $(14, \{67, 68, 70, 71, 74, 76 \}), 
		(15, \{68, 69, 71, 72, 75, 77 \}), 
		(16, \{69, 70, 72, 73, 76, 78 \}), 
		(17, \{66, 70, 71, 73, 74, 77 \}), 
		(18, \{67, 71, 72\}),$\\ 
		
		&& $(18, \{ 74, 75, 78 \}),
		(19, \{66, 68, 72, 73, 75, 76 \}), 
		(20, \{67, 69, 73, 74, 76, 77 \}), 
		(21, \{68, 70, 74, 75, 77, 78 \}), 
		(22, \{66, 69, 71, 75, 76, 78 \}),$\\ 
		
		&& $(23, \{66, 67, 70, 72, 76, 77 \}), 
		(24, \{67, 68, 71, 73, 77, 78 \}), 
		(25, \{66, 68, 69, 72, 74, 78 \}), 
		(26, \{66, 67, 69, 70, 73, 75 \})$\\
		
		$104$ & $(3,4)$ & $(27, \{41, 42, 43, 45, 47, 48, 49, 51 \}), 
		(28, \{42, 43, 44, 46, 48, 49, 50, 52 \}), 
		(29, \{40, 43, 44, 45, 47, 49, 50, 51 \}), 
		(30, \{41, 44, 45, 46\}),$\\
		
		&& $(30, \{ 48, 50, 51, 52 \}), 
		(31, \{40, 42, 45, 46, 47, 49, 51, 52 \}), 
		(32, \{40, 41, 43, 46, 47, 48, 50, 52 \}), 
		(33, \{40, 41, 42, 44, 47, 48, 49, 51 \}),$\\

		&& $(34, \{41, 42, 43, 45, 48, 49, 50, 52 \}), 
		(35, \{40, 42, 43, 44, 46, 49, 50, 51 \}), 
		(36, \{41, 43, 44, 45, 47, 50, 51, 52 \}),$\\ 
		
		&& $(37, \{40, 42, 44, 45, 46, 48, 51, 52 \}),
		(38, \{40, 41, 43, 45, 46, 47, 49, 52 \}), 
		(39, \{40, 41, 42, 44, 46, 47, 48, 50 \})$\\
		
		$78$ & $(3,5)$ & $(27, \{54, 55, 57, 58, 61, 63 \}), 
		(28, \{55, 56, 58, 59, 62, 64 \}), 
		(29, \{56, 57, 59, 60, 63, 65 \}), 
		(30, \{53, 57, 58, 60, 61, 64 \}), 
		(31, \{54, 58, 59\}),$\\
		
		&& $(31, \{61, 62, 65 \}), 
		(32, \{53, 55, 59, 60, 62, 63 \}), 
		(33, \{54, 56, 60, 61, 63, 64 \}), 
		(34, \{55, 57, 61, 62, 64, 65 \}), 
		(35, \{53, 56, 58, 62, 63, 65 \}),$\\ 
		
		&&$(36, \{53, 54, 57, 59, 63, 64 \}), 
		(37, \{54, 55, 58, 60, 64, 65 \}), 
		(38, \{53, 55, 56, 59, 61, 65 \}), 
		(39, \{53, 54, 56, 57, 60, 62 \}),$\\
		
		$91$ & $(3,6)$ &$(27, \{67, 68, 69, 72, 73, 74, 78 \}), 
		(28, \{66, 68, 69, 70, 73, 74, 75 \}), 
		(29, \{67, 69, 70, 71, 74, 75, 76 \}),
		(30, \{68, 70, 71, 72, 75, 76, 77 \}),$\\
		
		&&$(31, \{69, 71, 72, 73, 76, 77, 78 \}), 
		(32, \{66, 70, 72, 73, 74, 77, 78 \}), 
		(33, \{66, 67, 71, 73, 74, 75, 78 \}), 
		(34, \{66, 67, 68, 72, 74, 75, 76 \}),$\\ 
		
		&&$(35, \{67, 68, 69, 73, 75, 76, 77 \}), 
		(36, \{68, 69, 70, 74, 76, 77, 78 \}), 
		(37, \{66, 69, 70, 71, 75, 77, 78 \}),$\\ 
		
		&&$(38, \{66, 67, 70, 71, 72, 76, 78 \}), 
		(39, \{66, 67, 68, 71, 72, 73, 77 \})$\\
		
		$104$ & $(4,5)$ & 
		$(40, \{55, 57, 58, 59, 61, 63, 64, 65 \}), 
		(41, \{53, 56, 58, 59, 60, 62, 64, 65 \}), 
		(42, \{53, 54, 57, 59, 60, 61, 63, 65 \}), 
		(43, \{53, 54, 55, 58\}),$\\
		
		&& $(43, \{60, 61, 62, 64 \}), 
		(44, \{54, 55, 56, 59, 61, 62, 63, 65 \}), 
		(45, \{53, 55, 56, 57, 60, 62, 63, 64 \}), 
		(46, \{54, 56, 57, 58, 61, 63, 64, 65 \}),$\\
		
		&& $(47, \{53, 55, 57, 58, 59, 62, 64, 65 \}), 
		(48, \{53, 54, 56, 58, 59, 60, 63, 65 \}), 
		(49, \{53, 54, 55, 57, 59, 60, 61, 64 \}),$\\
		
		&&$ (50, \{54, 55, 56, 58, 60, 61, 62, 65 \}), 
		(51, \{53, 55, 56, 57, 59, 61, 62, 63 \}), 
		(52, \{54, 56, 57, 58, 60, 62, 63, 64 \})$\\
		
		$78$& $(4,6)$ & 
		$(40, \{69, 71, 74, 75, 77, 78 \}), 
		(41, \{66, 70, 72, 75, 76, 78 \}), 
		(42, \{66, 67, 71, 73, 76, 77 \}), 
		(43, \{67, 68, 72, 74, 77, 78 \}), 
		(44, \{66, 68, 69\},$\\
		
		&& $(44, \{73, 75, 78 \}),
		(45, \{66, 67, 69, 70, 74, 76 \}), 
		(46, \{67, 68, 70, 71, 75, 77 \}), 
		(47, \{68, 69, 71, 72, 76, 78 \}), 
		(48, \{66, 69, 70, 72, 73, 77 \}),$\\
		
		&& $(49, \{67, 70, 71, 73, 74, 78 \}), 
		(50, \{66, 68, 71, 72, 74, 75 \}), 
		(51, \{67, 69, 72, 73, 75, 76 \}), 
		(52, \{68, 70, 73, 74, 76, 77 \})$\\
		
		$104$ & $(5,6)$ & $(53, \{67, 68, 69, 71, 73, 74, 75, 77 \}), 
		(54, \{68, 69, 70, 72, 74, 75, 76, 78 \}), 
		(55, \{66, 69, 70, 71, 73, 75, 76, 77 \}), 
		(56, \{67, 70, 71, 72\}),$\\
		
		&& $(56, \{74, 76, 77, 78 \}), 
		(57, \{66, 68, 71, 72, 73, 75, 77, 78 \}), 
		(58, \{66, 67, 69, 72, 73, 74, 76, 78 \}), 
		(59, \{66, 67, 68, 70, 73, 74, 75, 77 \}),$\\
		
		&& $(60, \{67, 68, 69, 71, 74, 75, 76, 78 \}), 
		(61, \{66, 68, 69, 70, 72, 75, 76, 77 \}), 
		(62, \{67, 69, 70, 71, 73, 76, 77, 78 \}),$\\
		
		&& $(63, \{66, 68, 70, 71, 72, 74, 77, 78 \}), 
		(64, \{66, 67, 69, 71, 72, 73, 75, 78 \}), 
		(65, \{66, 67, 68, 70, 72, 73, 74, 76 \})$\\
		\bottomrule 
	\end{tabular}}
	\label{graph:78}
\end{table}
%%%% End of G_78

%%%%%%%%%%%%%%%%%%%%%%%%%%%%%%%%%%%%%%%
%%% Start of G_90
%%% Graph edge table for 90
\begin{table}
	\caption{The $1\,890$ edges of $G_{90}$ in Proposition\ref{prop90}. Vertices are partitioned into $10$ layers $\mathcal{L}_1, \ldots, \mathcal{L}_{10}$, with $9$ vertices in each layer. Vertices $z,z+1,\ldots,z+8$ are in $\mathcal{L}_{k+1}$ for $0 \leq k \leq 9$ and $z=9k+1$.}
	\setlength{\tabcolsep}{2pt}
	\renewcommand{\arraystretch}{1.1}
	\centering
	\resizebox{\textwidth}{!}{
	\begin{tabular}{c l | c l | c l}
		\toprule
		$i$ & \multicolumn{1}{c|}{$\left\{j : (i,j) \in E_{G_{90}}\right\}$} & 
		$i$ & \multicolumn{1}{c|}{$\left\{j : (i,j) \in E_{G_{90}}\right\}$} & 
		$i$ & \multicolumn{1}{c}{$\left\{j : (i,j) \in E_{G_{90}}\right\}$} \\
		\midrule
		$1$ & $\{ 
		2, 9, 10, 11, 12, 14, 15, 18, 24, 25, 30, 32, 33, 34, 36, 37, 38, 39, 41,44, 46, 51,$ & 
		$29$ & $\{30, 37, 38, 39, 42, 43, 45, 50, 51, 57,$ &
		$49$ & $\{50, 56, 57, 58, 59, 62, 63, 70, 71, 74 ,$ \\ 
		
		&$53, 54, 55, 57, 60, 62, 63, 66, 68, 69, 70, 72, 76, 77, 82, 83, 84, 86, 87, 90 \}$ & 
		&  $59, 60, 61, 63, 64, 65, 67, 70, 72, 74,$ & 
		& $ 77, 79, 80, 81, 83, 84, 85, 87, 90\}$\\
		
		$2$ & $\{ 
		3, 10, 11, 12, 13, 15, 16, 25, 26, 28, 31, 33, 34, 35, 38, 39, 40, 42, 
		45, 46, 47,$ & 
		& $75, 76, 78, 83, 84, 85, 87, 90\}$ & 
		$50$ & $\{51, 55, 57, 58, 59, 60, 63, 71, 72, 73 ,$\\
		
		&$52, 54, 55, 56, 58, 61, 63, 64, 67, 69, 70, 71, 77, 78, 82, 83, 84, 
		85, 87, 88 \}$ & 
		$30$ & $\{31, 37, 38, 39, 40, 43, 44, 51, 52, 55,$ & 
		& $ 75, 78, 80, 81, 82, 84, 85, 86, 88\}$\\
		
		$3$ & $\{ 
		4, 11, 12, 13, 14, 16, 17, 26, 27, 29, 32, 34, 35, 36, 37, 39, 40, 41, 
		43, 46, 47,$ & & 
		$58, 60, 61, 62, 64, 65, 66, 68, 71, 75,$ & 
		$51$ & $\{ 52, 55, 56, 58, 59, 60, 61, 64, 72, 73,$\\
		
		&$48, 53, 55, 56, 57, 59, 62, 65, 68, 70, 71, 72, 78, 79, 83, 84, 85, 
		86, 88, 89 \}$ &
		& $76, 77, 79, 82, 84, 85, 86, 88 \}$ & 
		& $\{ 74, 76, 79, 81, 83, 85, 86, 87, 89,$\\
		
		$4$ & $\{ 
		5, 12, 13, 14, 15, 17, 18, 19, 27, 28, 30, 33, 35, 36, 38, 40, 41, 42, 44, 47, 48,$ &
		$31$ & $\{32, 38, 39, 40, 41, 44, 45, 52, 53, 56,$ &
		$52$ & $\{53, 56, 57, 59, 60, 61, 62, 64, 65, 73 ,$\\
		
		&$49, 54, 56, 57, 58, 60, 63, 64, 66, 69, 71, 72, 79, 80, 84, 85, 86, 87, 89, 90 \}$ & 
		& $59, 61, 62, 63, 65, 66, 67, 69, 72, 76,$ &
		& $ 74, 75, 77, 80, 84, 86, 87, 88, 90\}$\\
		
		$5$ & $\{ 
		6, 10, 13, 14, 15, 16, 18, 19, 20, 28, 29, 31, 34, 36, 39, 41, 42, 43, 
		45, 46, 48,$ & 
		& $77, 78, 80, 83, 85, 86, 87, 89 \}$ & 
		$53$ & $\{54, 57, 58, 60, 61, 62, 63, 65, 66, 74,$\\
		
		&$49, 50, 55, 57, 58, 59, 61, 64, 65, 67, 70, 72, 80, 81, 82, 85, 86, 
		87, 88, 90 \}$ & 
		$32$ & $\{33, 37, 39, 40, 41, 42, 45, 53, 54, 55,$ &
		& $75, 76, 78, 81, 82, 85, 87, 88, 89 \}$\\
		
		$6$ & $\{ 
		7, 10, 11, 14, 15, 16, 17, 20, 21, 28, 29, 30, 32, 35, 37, 40, 42, 43, 
		44, 47, 49,$ & 
		& $57, 60, 62, 63, 64, 66, 67, 68, 70, 77,$ &
		$54$ & $\{55, 58, 59, 61, 62, 63, 66, 67, 73 ,$ \\
		
		&$50, 51, 56, 58, 59, 60, 62, 64, 65, 66, 68, 71, 73, 81, 82, 83, 86, 
		87, 88, 89 \}$ & 
		& $78, 79, 81, 84, 86, 87, 88, 90 \}$ &
		& $ 75, 76, 77, 79, 83, 86, 88, 89, 90\}$\\
		
		$7$ & $\{ 
		8, 11, 12, 15, 16, 17, 18, 21, 22, 29, 30, 31, 33, 36, 38, 41, 43, 44, 
		45, 48, 50,$ & 
		$33$ & $\{34, 37, 38, 40, 41, 42, 43, 46, 54, 55,$ &
		$55$ & $\{56, 63, 64, 65, 66, 68, 69, 72,$ \\
		
		&$51, 52, 57, 59, 60, 61, 63, 65, 66, 67, 69, 72, 73, 74, 83, 84, 87, 
		88, 89, 90 \}$ & 
		& $56, 58, 61, 63, 65, 67, 68, 69, 71, 73,$ &
		& $78, 79, 84, 86, 87, 88, 90 \}$\\
		
		$8$ & $\{ 
		9, 10, 12, 13, 16, 17, 18, 22, 23, 28, 30, 31, 32, 34, 37, 39, 42, 44, 
		45, 49, 51,$ & 
		& $78, 79, 80, 82, 85, 87, 88, 89 \}$ & 
		$56$ & $\{ 57, 64, 65, 66, 67, 69, 70,$\\
		
		&$52, 53, 55, 58, 60, 61, 62, 64, 66, 67, 68, 70, 74, 75, 82, 84, 85, 
		88, 89, 90\}$ & 
		$34$ & $\{35, 38, 39, 41, 42, 43, 44, 46, 47, 55,$ &
		& $ 79, 80, 82, 85, 87, 88, 89\}$\\
		
		$9$ & $\{ 
		10, 11, 13, 14, 17, 18, 23, 24, 29, 31, 32, 33, 35, 37, 38, 40, 43, 45, 50, 52,$ & 
		& $56, 57, 59, 62, 66, 68, 69, 70, 72, 74,$ &
		$57$ & $\{ 58, 65, 66, 67, 68, 70, 71,$\\
		
		&$53, 54, 56, 59, 61, 62, 63, 65, 67, 68, 69, 71, 75, 76, 82, 83, 85, 86, 89, 90 \}$ & 
		& $79, 80, 81, 83, 86, 88, 89, 90 \}$ &
		& $80, 81, 83, 86, 88, 89, 90\}$\\
		
		$10$ & $\{ 
		11, 18, 19, 20, 23, 24, 26, 27, 31, 32, 38, 40, 41, 42, 44, 46, 48,51,$ & 
		$35$ &$\{36, 39, 40, 42, 43, 44, 45, 47, 48, 56,$ &
		$58$ & $\{59, 66, 67, 68, 69, 71, 72,$\\
		
		&$53, 54, 55, 56, 57, 59, 64, 65, 66, 68, 71, 74, 76, 77, 78, 80, 87, 88 \}$ & 
		& $57, 58, 60, 63, 64, 67, 69, 70, 71, 73,$ & 
		& $73, 81, 82, 84, 87, 89, 90 \}$\\
		
		$11$ & $\{ 
		12, 19, 20, 21, 24, 25, 27, 32, 33, 39, 41, 42, 43, 45, 46, 47, 49, 52,$ & 
		& $75, 80, 81, 82, 84, 87, 89, 90\}$ & 
		$59$ & $\{ 60, 64, 67, 68, 69, 70, 72,$\\
		
		&$54, 56, 57, 58, 60, 65, 66, 67, 69, 72, 75, 77, 78, 79, 81, 88, 89 \}$ & 
		$36$ & $\{ 37, 40, 41, 43, 44, 45, 48, 49, 55, 57,$ & 
		& $ 73, 74, 82, 83, 85, 88, 90\}$\\
		
		$12$ & $\{ 
		13, 19, 20, 21, 22, 25, 26, 33, 34, 37, 40, 42, 43, 44, 46, 47, 48, 50,$ & 
		& $58, 59, 61, 65, 68, 70, 71, 72, 73, 74,$ &
		$60$ & $\{61, 64, 65, 68, 69, 70, 71,$ \\
		
		&$53, 57, 58, 59, 61, 64, 66, 67, 68, 70, 73, 76, 78, 79, 80, 89, 90 \}$ & 
		& $76, 81, 82, 83, 85, 88, 90 \}$ &
		& $74, 75, 82, 83, 84, 86, 89 \}$\\
		
		$13$ & $\{ 
		14, 20, 21, 22, 23, 26, 27, 34, 35, 38, 41, 43, 44, 45, 47, 48, 49, 51,$ & 
		$37$ & $\{38, 45, 46, 47, 48, 50, 51, 54,$ &
		$61$ & $\{ 62, 65, 66, 69, 70, 71, 72,$ \\
		
		&$54, 58, 59, 60, 62, 65, 67, 68, 69, 71, 74, 77, 79, 80, 81, 82, 90 \}$ & 
		& $60, 61, 66, 68, 69, 70, 72, 73,$ &
		& $75, 76, 83, 84, 85, 87, 90 \}$ \\
		
		$14$ & $\{ 
		15, 19, 21, 22, 23, 24, 27, 35, 36, 37, 39, 42, 44, 45, 46, 48, 49, 50,$ & 
		& $ 74, 75, 77, 80, 82, 87, 89, 90\}$ &
		$62$ & $\{63, 64, 66, 67, 70, 71, 72,$ \\
		
		&$52, 59, 60, 61, 63, 66, 68, 69, 70, 72, 73, 75, 78, 80, 81, 82, 83 \}$ & 
		$38$ & $\{ 39, 46, 47, 48, 49, 51, 52, 61,$ &
		& $76, 77, 82, 84, 85, 86, 88 \}$ \\
		
		$15$ & $\{ 
		16, 19, 20, 22, 23, 24, 25, 28, 36, 37, 38, 40, 43, 45, 47, 49, 50, 51,$ & 
		& $62, 64, 67, 69, 70, 71, 74, 75,$ &
		$63$ & $\{ 64, 65, 67, 68, 71, 72, 77,$ \\
		
		&$53, 55, 60, 61, 62, 64, 67, 69, 70, 71, 73, 74, 76, 79, 81, 83, 84 \}$ & 
		& $76, 78, 81, 82, 83, 88, 90\}$ &
		& $78, 83, 85, 86, 87, 89\}$ \\
		
		$16$ & $\{ 
		17, 20, 21, 23, 24, 25, 26, 28, 29, 37, 38, 39, 41, 44, 48, 50, 51, 52,$ & 
		$39$ & $\{ 40, 47, 48, 49, 50, 52, 53, 62,$ &
		$64$ & $\{ 65, 72, 73, 74, 77, 78, 80, 81, 85, 86 \}$\\
		
		&$54, 56, 61, 62, 63, 65, 68, 70, 71, 72, 73, 74, 75, 77, 80, 84, 85 \}$ &
		& $ 63, 65, 68, 70, 71, 72, 73, 75,$ & 
		$65$ & $\{66, 73, 74, 75, 78, 79, 81, 86, 87\}$ \\ 
		
		$17$ & $\{ 
		18, 21, 22, 24, 25, 26, 27, 29, 30, 38, 39, 40, 42, 45, 46, 49, 51, 52,$ &
		& $76, 77, 79, 82, 83, 84, 89 \}$ & 
		$66$ & $\{67, 73, 74, 75, 76, 79, 80, 87, 88\}$ \\
		
		&$53, 55, 57, 62, 63, 64, 66, 69, 71, 72, 74, 75, 76, 78, 81, 85, 86 \}$ & 
		$40$ & $\{ 41, 48, 49, 50, 51, 53, 54, 55,$ &
		$67$ & $\{68, 74, 75, 76, 77, 80, 81, 88, 89\}$ \\ 
		
		$18$ & $\{ 
		19, 22, 23, 25, 26, 27, 30, 31, 37, 39, 40, 41, 43, 47, 50, 52, 53, 54,$ & 
		& $63, 64, 66, 69, 71, 72, 74, 76 ,$ & 
		$68$ & $\{69, 73, 75, 76, 77, 78, 81, 89, 90\}$ \\
		
		&$55,56, 58, 63, 64, 65, 67, 70, 72, 73, 75, 76, 77, 79, 86, 87 \}$ & 
		& $77, 78, 80, 83, 84, 85, 90 \}$ &
		$69$ & $\{70, 73, 74, 76, 77, 78, 79, 82, 90\}$ \\ 
		
		$19$ & $\{ 
		20, 27, 28, 29, 30, 32, 33, 36, 42, 43, 48, 50, 51, 52, 54, 55, 56, 57,$ & 
		$41$ & $\{42, 46, 49, 50, 51, 52, 54, 55 ,$ & 
		$70$ & $\{71, 74, 75, 77, 78, 79, 80, 82, 83\}$ \\ 
		
		&$59,62,64, 69, 71, 72, 73, 75, 78, 80, 81, 84, 86, 87, 88, 90 \}$ & & $ 56, 64, 65, 67, 70, 72, 75, 77,$ & 
		$71$ & $\{72, 75, 76, 78, 79, 80, 81, 83, 84\}$ \\
		
		$20$ & $\{ 
		21, 28, 29, 30, 31, 33, 34, 43, 44, 46, 49, 51, 52, 53, 56, 57, 58, 60,$ & 
		& $78, 79, 81, 82, 84, 85, 86 \}$ & 
		$72$ & $\{73, 76, 77, 79, 80, 81, 84, 85\}$ \\ 
		
		&$63, 64, 65, 70, 72, 73, 74, 76, 79, 81, 82, 85, 87, 88, 89\}$ & 
		$42$ & $\{ 43, 46, 47, 50, 51, 52, 53, 56,$ &
		$73$ & $\{74, 81, 82, 83, 84, 86, 87, 90\}$ \\
		
		$21$ & $\{ 
		22, 29, 30, 31, 32, 34, 35, 44, 45, 47, 50, 52, 53, 54, 55, 57, 58, 59,$ & 
		& $ 57, 64, 65, 66, 68, 71, 73, 76,$ & 
		$74$ & $\{75, 82, 83, 84, 85, 87, 88\}$ \\
		
		&$61, 64, 65, 66, 71, 73, 74, 75, 77, 80, 83, 86, 88, 89, 90 \}$ & 
		& $78, 79, 80, 83, 85, 86, 87 \}$ &
		$75$ & $\{76, 83, 84, 85, 86, 88, 89\}$ \\ 
		
		$22$ & $\{ 
		23, 30, 31, 32, 33, 35, 36, 37, 45, 46, 48, 51, 53, 54, 56, 58, 59, 60,$ & 
		$43$ & $\{ 44, 47, 48, 51, 52, 53, 54, 57,$ & 
		$76$ & $\{77, 84, 85, 86, 87, 89, 90\}$ \\ 
		
		&$62, 65, 66, 67, 72, 74, 75, 76, 78, 81, 82, 84, 87, 89, 90 \}$ & 
		& $ 58, 65, 66, 67, 69, 72, 74, 77,$ & 
		$77$ & $\{78, 82, 85, 86, 87, 88, 90\}$ \\
		
		$23$ & $\{ 
		24, 28, 31, 32, 33, 34, 36, 37, 38, 46, 47, 49, 52, 54, 57, 59, 60, 61,$ &
		& $79, 80, 81, 84, 86, 87, 88 \}$ & 
		$78$ & $\{79, 82, 83, 86, 87, 88, 89\}$ \\
		
		&$63, 64, 66, 67, 68, 73, 75, 76, 77, 79, 82, 83, 85, 88, 90 \}$ & 
		$44$ & $\{ 45, 46, 48, 49, 52, 53, 54, 58,$ &
		$79$ & $\{80, 83, 84, 87, 88, 89, 90\}$ \\
		
		$24$ & $\{ 
		25, 28, 29, 32, 33, 34, 35, 38, 39, 46, 47, 48, 50, 53, 55, 58, 60, 61,$ &
		& $ 59, 64, 66, 67, 68, 70, 73, 75,$ & 
		$80$ & $\{81, 82, 84, 85, 88, 89, 90\}$ \\ 
		
		&$ 62, 65, 67, 68, 69, 74, 76, 77, 78, 80, 82, 83, 84, 86, 89\}$ & 
		& $78, 80, 81, 85, 87, 88, 89 \}$ & 
		$81$ & $\{82, 83, 85, 86, 89, 90\}$ \\
		
		$25$ & $\{ 
		26, 29, 30, 33, 34, 35, 36, 39, 40, 47, 48, 49, 51, 54, 56, 59, 61, 62,$ & 
		$45$ & $\{ 46, 47, 49, 50, 53, 54, 59, 60,$  &
		$82$ & $\{83, 90\}$ \\
		
		&$ 63, 66, 68, 69, 70, 75, 77, 78, 79, 81, 83, 84, 85, 87, 90\}$ & 
		& $65, 67, 68, 69, 71, 73, 74, 76 ,$ & 
		$83$ & $\{84\}$ \\ 
		
		$26$ & $\{ 
		27, 28, 30, 31, 34, 35, 36, 40, 41, 46, 48, 49, 50, 52, 55, 57, 60, 62,$ & 
		& $ 79, 81, 86, 88, 89, 90\}$ & 
		$84$ & $\{85\}$ \\
		
		&$63, 67, 69, 70, 71, 73, 76, 78, 79, 80, 82, 84, 85, 86, 88 \}$ & 
		$46$ & $\{47, 54, 55, 56, 59, 60, 62, 63, 67, 68,$ &
		$85$ & $\{86\}$ \\
		
		$27$ & $\{ 
		28, 29, 31, 32, 35, 36, 41, 42, 47, 49, 50, 51, 53, 55, 56, 58, 61, 63,$ & 
		& $ 74, 76, 77, 78, 80, 82, 84, 87, 89, 90\}$ &
		$86$ & $\{87\}$ \\
		
		&$ 68, 70, 71, 72, 74, 77, 79, 80, 81, 83, 85, 86, 87, 89\}$ & 
		$47$ & $\{48, 55, 56, 57, 60, 61, 63, 68, 69, 75,$ &
		$87$ & $\{88\}$ \\ 
		
		$28$ & $\{ 
		29, 36, 37, 38, 41, 42, 44, 45, 49, 50, 56, 58, 59, 60, 62,$ & 
		& $77, 78, 79, 81, 82, 83, 85, 88, 90\}$ & 
		$88$ & $\{89\}$ \\
		
		&$ 64, 66, 69, 71, 72, 73, 74, 75, 77, 82, 83, 84, 86, 89\}$ & 
		$48$ & $\{49, 55, 56, 57, 58, 61, 62, 69, 70, 73,$ & 
		$89$ & $\{ 90\}$ \\
		
		&& 
		& $76, 78, 79, 80, 82, 83, 84, 86, 89\}$ &
		&\\ 
		\bottomrule 
	\end{tabular}
	}
	\label{graph:90} 
\end{table}

%%%%%%%%% End of G_90

%%% Start of new G_91
\begin{table}[!htbp]
	\caption{The $2\,002$ Edges of $G_{91}$ in Proposition \ref{prop91}: Vertices are partitioned into $7$ layers $\mathcal{L}_1, \ldots,\mathcal{L}_7$, with $13$ vertices in each layer. Vertices $z, z+1, \ldots, z+12$ are in $\mathcal{L}_{k+1}$ for $0 \leq k \leq 6$ and $z=13k+1$.}
	\setlength{\tabcolsep}{2pt}
	\renewcommand{\arraystretch}{1.1}
	\centering
	\resizebox{\textwidth}{!}{
	\begin{tabular}{c l | c l | c l}
		\toprule
		$i$ & \multicolumn{1}{c|}{$\left\{j : (i,j) \in E_{G_{91}}\right\}$} & 
		$i$ & \multicolumn{1}{c|}{$\left\{j : (i,j) \in E_{G_{91}}\right\}$} & 
		$i$ & \multicolumn{1}{c}{$\left\{j : (i,j) \in E_{G_{91}}\right\}$} \\
		\midrule
		
		$1$ & $\{ 
		4, 5, 10, 12, 13, 14, 19, 20, 21, 24, 26, 27, 28, 30, 31, 32, 34, 40, 41,42, 47, 48, $ & $31$ & $\{34, 35, 36, 37, 40, 42, 43, 44, 49, 50,$  &
		$51$ & $\{54, 55, 57, 60, 62, 63, 64, 69, 70, 71,$ \\
		
		& $49, 51,52, 53, 58, 59, 60, 66, 68, 69, 70,
		72, 73, 74, 76, 79, 80, 81, 86, 87, 88, 90 \}$ &
		& $51, 54, 56, 57, 60, 61, 62, 64, 65, 70,$ &  
		& $74, 76, 77, 80, 81, 82, 84, 85, 90, 91\}$\\
		
		$2$ & $\{ 
		5, 6, 8, 11, 13, 14, 15, 20, 21, 22, 25, 27, 28, 31, 32, 33, 35, 36, 41, 42, 43, 48,$ & 
		& $71, 72, 77, 81, 82, 83, 88, 89, 90\}$  &
		$52$ & $\{55, 56, 57, 58, 61, 63, 64, 65, 70, 71,$ \\
		
		& $49, 52, 53, 54, 59, 60, 61, 64, 67, 69, 70, 73, 74, 75, 77, 80, 81, 82, 87, 88, 89, 91\}$& 
		$32$ & $\{35, 36, 37, 38, 41, 43, 44, 45, 50, 51,$ & 
		& $72, 75, 77, 78, 81, 82, 83, 85, 86, 91\}$ \\		
		
		$3$ & $\{ 
		6, 7, 8, 9, 12, 14, 15, 16, 21, 22, 23, 26, 28, 29, 32, 33, 34, 36, 37, 
		42, 43, 44, 49,$ & 
		& $52, 55, 58, 61, 62, 63, 64, 65, 66, 71,$ &
		$53$ & $\{56, 57, 58, 59, 62, 64, 65, 66, 71, 72,$ \\
		
		& $53, 54, 55, 60, 61, 62, 64, 65, 68, 70, 71, 74, 75, 76, 81, 82, 83, 85, 88, 89, 90 \}$& & $72, 73, 82, 83, 84, 89, 90, 91 \}$&
		& $73, 76, 79, 82, 83, 84, 85, 86, 87\}$\\		
		
		$4$ & $\{ 
		7, 8, 9, 10, 13, 15, 16, 17, 22, 23, 24, 27, 30, 33, 34, 35, 36, 37, 38, 
		43, 44, 45,$ & 
		$33$ & $\{37, 38, 39, 42, 44, 45, 46, 51, 52, 53,$ &
		$54$ & $\{58, 59, 60, 63, 65, 66, 67, 72, 73,$ \\
		
		& $54, 55, 56, 61, 62, 63, 64, 65, 66, 69, 72, 75, 76, 77, 82, 83, 84, 
		86, 89, 90, 91 \}$ & 
		& $56, 57, 59, 62, 63, 65, 66, 67, 72, 73,$ &
		& $74, 77, 78, 80, 83, 84, 86, 87, 88 \}$\\		
		
		$5$ & $\{ 
		9, 10, 11, 14, 16, 17, 18, 23, 24, 25, 28, 29, 31, 34, 35, 37, 38, 39, 
		44, 45, 46,$ & 
		& $74, 78, 83, 84, 85, 90, 91\}$  &
		$55$ & $\{57, 59, 60, 61, 66, 67, 68, 71, 73, $ \\
		
		& $50, 55, 56, 57, 62, 63, 65, 66, 67, 70, 71, 73, 76, 77, 78, 83, 84, 
		85, 87, 90, 91 \}$ &
		$34$ & $\{36, 38, 39, 40, 45, 46, 47, 50, 52,$ &
		& $74, 75, 78, 79, 81, 84, 87, 88, 89\}$\\		
		
		$6$ & $\{ 
		8, 10, 11, 12, 17, 18, 19, 22, 24, 25, 26, 29, 30, 32, 35, 38, 39, 40, 
		45, 46, 47, $& 
		& $53, 54, 57, 58, 60, 63, 66, 67, 68,$  &
		$56$ & $\{58, 60, 61, 62, 67, 68, 69, 72, 74,$ \\
		
		& $50, 51, 56, 57, 58, 63, 64, 66, 67, 68, 71, 72, 74, 77, 78, 79, 84, 
		85, 86, 88, 91 \}$&
		& $73, 74, 75, 78, 79, 84, 85, 86, 91\}$ &
		& $75, 76, 78, 79, 80, 82, 88, 89, 90\}$\\
		
		$7$ & $\{ 
		9, 11, 12, 13, 18, 19, 20, 23, 25, 26, 27, 29, 30, 31, 33, 39, 40, 41, 
		46, 47, 48,$ & 
		$35$& $\{37, 39, 40, 41, 46, 47, 48, 51, 53,$  &
		$57$ & $\{60, 61, 66, 68, 69, 70, 75, 76, 77,$ \\
		
		& $50, 51, 52, 57, 58, 59, 65, 67, 68, 69, 71, 72, 73, 75, 78, 79, 80, 
		85, 86, 87, 89 \}$&
		& $54, 55, 57, 58, 59, 61, 67, 68, 69,$ &
		& $80, 82, 83, 84, 86, 87, 88, 90\}$\\		
		
		$8$ & $\{ 
		11, 12, 17, 19, 20, 21, 26, 27, 28, 31, 33, 34, 35, 37, 38, 39, 41, 47, 
		48, 49,$& 
		& $74, 75, 76, 78, 79, 80, 85, 86, 87\}$  &
		$58$ & $\{61, 62, 64, 67, 69, 70, 71, 76, 77,$ \\
		
		& $54, 55, 56, 58, 59, 60, 65, 66, 67, 73, 75, 76, 77, 79, 80, 81, 83, 86, 87, 88 \}$ &
		$36$ & $\{39, 40, 45, 47, 48, 49, 54, 55, 56,$ & 
		& $78, 81, 83, 84, 87, 88, 89, 91\}$\\	
		
		$9$ & $\{ 
		12, 13, 15, 18, 20, 21, 22, 27, 28, 29, 32, 34, 35, 38, 39, 40, 42, 43, 
		48, 49,$ & 
		& $59, 61, 62, 63, 65, 66, 67, 69, 75,$  &
		$59$ & $\{62, 63, 64, 65, 68, 70, 71, 72, 77,$ \\
		
		& 
		$50, 55, 56, 59, 60, 61, 66, 67, 68, 71, 74, 76, 77, 80, 81, 82, 84, 87, 
		88, 89 \}$&
		& $76, 77, 82, 83, 84, 86, 87, 88\}$ & & $78, 79, 82, 84, 85, 88, 89, 90\}$\\	
		
		$10$ & $\{ 
		13, 14, 15, 16, 19, 21, 22, 23, 28, 29, 30, 33, 35, 36, 39, 40, 41, 43, 
		44, 49, $& 
		$37$& $\{40, 41, 43, 46, 48, 49, 50, 55, 56,$  &
		$60$ & $\{63, 64, 65, 66, 69, 71, 72, 73,$ \\
		
		& $50, 51, 56, 60, 61, 62, 67, 68, 69, 71, 72, 75, 77, 78, 81, 82, 83, 88, 89, 90 \}$
		& & $57, 60, 62, 63, 66, 67, 68, 70, 71,$ & 
		& $78, 79, 80, 83, 86, 89, 90, 91\}$\\		
		
		$11$ & $\{ 
		14, 15, 16, 17, 20, 22, 23, 24, 29, 30, 31, 34, 37, 40, 41, 42, 43, 44, 
		45, 50,$ & 
		& $76, 77, 78, 83, 84, 87, 88, 89\}$  &
		$61$ & $\{65, 66, 67, 70, 72, 73, 74, 79,$ \\
		
		& $51, 52, 61, 62, 63, 68, 69, 70, 71, 72, 73, 76, 79, 82, 83, 84, 89, 90, 91 \}$&
		$38$& $\{41, 42, 43, 44, 47, 49, 50, 51, 56,$ &
		& $80, 81, 84, 85, 87, 90, 91\}$\\		
		
		$12$ & $\{ 
		16, 17, 18, 21, 23, 24, 25, 30, 31, 32, 35, 36, 38, 41, 42, 44, 45, 46, 
		51,$ & 
		& $57, 58, 61, 63, 64, 67, 68, 69, 71,$  &
		$62$ & $\{64, 66, 67, 68, 73, 74, 75, 78,$ \\
		
		& $52, 53, 57, 62, 63, 64, 69, 70, 72, 73, 74, 77, 78, 80, 83, 84, 85, 90, 91 \}$&
		& $72, 77, 78, 79, 84, 88, 89, 90 \}$ & 
		& $80, 81, 82, 85, 86, 88, 91\}$\\		
		
		$13$ & $\{ 
		15, 17, 18, 19, 24, 25, 26, 29, 31, 32, 33, 36, 37, 39, 42, 45, 46, 47, 
		52,$ & 
		$39$ & $\{42, 43, 44, 45, 48, 50, 51, 52, 57,$  &
		$63$ & $\{65, 67, 68, 69, 74, 75, 76, 79,$ \\
		
		& $53, 54, 57, 58, 63, 64, 65, 70, 71, 73, 74, 75, 78, 79, 81, 84, 85, 86, 91 \}$&
		& $58, 59, 62, 65, 68, 69, 70, 71,$ &
		& $81, 82, 83, 85, 86, 87, 89\}$ \\		
		
		$14$ & $\{ 
		16, 18, 19, 20, 25, 26, 27, 30, 32, 33, 34, 36, 37, 38, 40, 46, 47, 48, 
		53,$& 
		& $72, 73, 78, 79, 80, 89, 90, 91 \}$  &
		$64$ & $\{67, 68, 73, 75, 76, 77, 82, 83,$\\
		
		& $54, 55, 57, 58, 59, 64, 65, 66, 72, 74, 75, 76, 78, 79, 80, 82, 85, 86, 87 \}$&
		$40$& $\{44, 45, 46, 49, 51, 52, 53, 58, 59,$ &
		& $84, 87, 89, 90, 91\}$\\		
		
		$15$ & $\{ 
		18, 19, 24, 26, 27, 28, 33, 34, 35, 38, 40, 41, 42, 44, 45, 46, 48, 54, 
		55,$ & 
		& $60, 63, 64, 66, 69, 70, 72, 73,$ &
		$65$ & $\{68, 69, 71, 74, 76, 77, 78,$\\
		
		& $56, 61, 62, 63, 65, 66, 67, 72, 73, 74, 80, 82, 83, 84, 86, 87, 88, 90 \}$& & 
		$74, 79, 80, 81, 85, 90, 91\}$&& $83, 84, 85, 88, 90, 91\}$\\		
		
		$16$ & $\{ 
		19, 20, 22, 25, 27, 28, 29, 34, 35, 36, 39, 41, 42, 45, 46, 47, 49, 50, 
		55,$& 
		$41$ & $\{43, 45, 46, 47, 52, 53, 54, 57,$  &
		$66$ & $\{69, 70, 71, 72, 75, 77, 78,$ \\
		
		& $56, 57, 62, 63, 66, 67, 68, 73, 74, 75, 78, 81, 83, 84, 87, 88, 89, 91 \}$&
		& $59, 60, 61, 64, 65, 67, 70, 73,$&& $79, 84, 85, 86, 89, 91\}$\\		
		
		$17$ & $\{ 
		20, 21, 22, 23, 26, 28, 29, 30, 35, 36, 37, 40, 42, 43, 46, 47, 48, 50, 
		51,$ & 
		& $ 74, 75, 80, 81, 82, 85, 86, 91\}$ &
		$67$ & $\{70, 71, 72, 73, 76, 78,$\\
		
		& $56, 57, 58, 63, 67, 68, 69, 74, 75, 76, 78, 79, 82, 84, 85, 88, 89, 90 \}$& 
		$42$& $\{ 44, 46, 47, 48, 53, 54, 55, 58,$&& $79, 80, 85, 86, 87, 90\}$\\		
		
		$18$ & $\{ 
		21, 22, 23, 24, 27, 29, 30, 31, 36, 37, 38, 41, 44, 47, 48, 49, 50, 51,$ & 
		& $60, 61, 62, 64, 65, 66, 68, 74,$ &
		$68$ & $\{72, 73, 74, 77, 79, 80,$ \\
		
		& $52, 57, 58, 59, 68, 69, 70, 75, 76, 77, 78, 79, 80, 83, 86, 89, 90, 91\}$& 
		& $ 75, 76, 81, 82, 83, 85, 86, 87\}$& 
		& $81, 86, 87, 88, 91\}$\\		
		
		$19$ & $\{ 
		23, 24, 25, 28, 30, 31, 32, 37, 38, 39, 42, 43, 45, 48, 49, 51, 52, 53,$& 
		$43$ & $\{46, 47, 52, 54, 55, 56, 61, 62,$  &
		$69$ & $\{71, 73, 74, 75, 80, 81,$ \\
		
		& $58, 59, 60, 64, 69, 70, 71, 76, 77, 79, 80, 81, 84, 85, 87, 90, 91 \}$&
		& $63, 66, 68, 69, 70, 72, 73, 74,$ && $82, 85, 87, 88, 89\}$\\		
		
		$20$ & $\{ 
		22, 24, 25, 26, 31, 32, 33, 36, 38, 39, 40, 43, 44, 46, 49, 52, 53, 54,$& 
		& $ 76, 82, 83, 84, 89, 90, 91\}$ &
		$70$ & $\{72, 74, 75, 76, 81, 82,$ \\
		
		& $59, 60, 61, 64, 65, 70, 71, 72, 77, 78, 80, 81, 82, 85, 86, 88, 91 \}$ &
		$44$ & $\{ 47, 48, 50, 53, 55, 56, 57, 62,$ && $83, 86, 88, 89, 90\}$\\		
		
		$21$ & $\{ 
		23, 25, 26, 27, 32, 33, 34, 37, 39, 40, 41, 43, 44, 45, 47, 53, 54, 55,$& 
		& $63, 64, 67, 69, 70, 73, 74, 75,$  &
		$71$ & $\{74, 75, 80, 82, 83, 84, 89, 90, 91\}$\\
		
		& $60, 61, 62, 64, 65, 66, 71, 72, 73, 79, 81, 82, 83, 85, 86, 87, 89 \}$ & 
		& $77, 78, 83, 84, 85, 90, 91\}$& 
		$72$ & $\{75, 76, 78, 81, 83, 84, 85, 90, 91\}$\\		
		
		$22$ & $\{ 
		25, 26, 31, 33, 34, 35, 40, 41, 42, 45, 47, 48, 49, 51, 52, 53, 55, 61,$& 
		$45$ & $\{48, 49, 50, 51, 54, 56, 57, 58, $ &
		$73$ & $\{76, 77, 78, 79, 82, 84, 85, 86, 91\}$\\	
		
		& $ 62, 63, 68, 69, 70, 72, 73, 74, 79, 80, 81, 87, 89, 90, 91\}$& 
		& $63, 64, 65, 68, 70, 71, 74, 75,$& 
		$74$ & $\{77, 78, 79, 80, 83, 85, 86, 87\}$\\		
		
		$23$ & $\{ 
		26, 27, 29, 32, 34, 35, 36, 41, 42, 43, 46, 48, 49, 52, 53, 54, 56,$& 
		& $76, 78, 79, 84, 85, 86, 91\}$ &
		$75$ & $\{79, 80, 81, 84, 86, 87, 88\}$\\	
		
		& $57, 62, 63, 64, 69, 70, 73, 74, 75, 80, 81, 82, 85, 88, 90, 91 \}$& 
		$46$ & $\{49, 50, 51, 52, 55, 57, 58, 59, $ & 
		$76$& $\{78, 80, 81, 82, 87, 88, 89\}$\\		
		
		$24$ & $\{ 
		27, 28, 29, 30, 33, 35, 36, 37, 42, 43, 44, 47, 49, 50, 53, 54, 55,$& 
		& $64, 65, 66, 69, 72, 75, 76,$  &
		$77$ & $\{79, 81, 82, 83, 88, 89, 90\}$\\	
		
		& $57, 58, 63, 64, 65, 70, 74, 75, 76, 81, 82, 83, 85, 86, 89, 91 \}$& 
		& $ 77, 78, 79, 80, 85, 86, 87\}$& $78$ & $\{81, 82, 87, 89, 90, 91\}$\\		
		
		$25$ & $\{ 
		28, 29, 30, 31, 34, 36, 37, 38, 43, 44, 45, 48, 51, 54, 55, 56,$& 
		$47$ & $\{51, 52, 53, 56, 58, 59, 60, 65,$  &
		$79$ & $\{82, 83, 85, 88, 90, 91\}$\\	
		
		& $57, 58, 59, 64, 65, 66, 75, 76, 77, 82, 83, 84, 85, 86, 87, 90\}$
		& 
		& $66, 67, 70, 71, 73, 76, 77,$ &
		$80$ & $\{83, 84, 85, 86, 89, 91\}$\\	
		
		$26$ & $\{ 
		30, 31, 32, 35, 37, 38, 39, 44, 45, 46, 49, 50, 52, 55, 56, 58,$& 
		& $ 79, 80, 81, 86, 87, 88\}$ &
		$81$ & $\{84, 85, 86, 87, 90\}$\\	
		
		& $ 59, 60, 65, 66, 67, 71, 76, 77, 78, 83, 84, 86, 87, 88, 91\}$& 
		$48$ & $\{50, 52, 53, 54, 59, 60, 61,$ & $82$ & $\{86, 87, 88, 91\}$\\		
		
		$27$ & $\{ 
		29, 31, 32, 33, 38, 39, 40, 43, 45, 46, 47, 50, 51, 53, 56, 59,$& 
		& $64, 66, 67, 68, 71, 72, 74,$ &
		$83$ & $\{85, 87, 88, 89\}$\\	
		
		& $60, 61, 66, 67, 68, 71, 72, 77, 78, 79, 84, 85, 87, 88, 89 \}$&
		& $ 77, 80, 81, 82, 87, 88, 89\}$ &
		$84$ & $\{86, 88, 89, 90\}$\\			
		
		$28$ & $\{ 
		30, 32, 33, 34, 39, 40, 41, 44, 46, 47, 48, 50, 51, 52, 54, 60,$& 
		$49$ & $\{51, 53, 54, 55, 60, 61, 62,$  &
		$85$ & $\{88, 89\}$\\	
		
		& $61, 62, 67, 68, 69, 71, 72, 73, 78, 79, 80, 86, 88, 89, 90 \}$& 
		& $65, 67, 68, 69, 71, 72, 73,$ &
		$86$ & $\{89, 90\}$\\			
		
		$29$ & $\{
		32, 33, 38, 40, 41, 42, 47, 48, 49, 52, 54, 55, 56, 58, 59,$ &
		& $ 75, 81, 82, 83, 88, 89, 90\}$ &
		$87$ & $\{90, 91\}$\\	
		
		& $60, 62, 68, 69, 70,75, 76, 77, 79, 80, 81, 86, 87, 88\}$ &
		$50$ & $\{53, 54, 59, 61, 62, 63, 68, 69, 70, 73,$ & 
		$88$ & $\{91\}$\\	
		
		$30$ & $\{
		33, 34, 36, 39, 41, 42, 43, 48, 49, 50, 53, 55, 56, 59, 60,$ &
		& $75, 76, 77, 79, 80, 81, 83, 89, 90, 91\}$ &
		& \\	
		
		& $ 61, 63, 64, 69, 70,71, 76, 77, 80, 81, 82, 87, 88, 89\}$ &
		& \\	
		\bottomrule 
	\end{tabular}
	}
	\label{table:List91} 
\end{table}

%%%% End of new G_91

%%%%%%% Start of G_93
\begin{table}
	\caption{The $1\,302$ edges of $G_{93}$ in Proposition\ref{prop93}. Vertices are partitioned into $3$ layers of $31$ vertices each. Vertices $1$ to $31$ are in $\mathcal{L}_1$. Vertices $32$ to $62$ are in $\mathcal{L}_2$. Vertices $63$ to $93$ are in $\mathcal{L}_3$. Edges $\{(i,j)\}$ are listed as $(i,\{j \in J\})$ with $i \in \mathcal{L}_r$ and $J \subset \mathcal{L}_s$}
	\setlength{\tabcolsep}{3pt}
	\renewcommand{\arraystretch}{1.1}
	\centering
	\resizebox{\textwidth}{!}{
	\begin{tabular}{ccl}
		\toprule
		$\#$ & $(r,s)$ & \multicolumn{1}{c}{$(i,\{j \in J\})$} \\
		\midrule
		%% Layer 1
		$124$ & $(1,1)$ & 
		$(1, \{ 11, 13, 15, 16, 17, 18, 20, 22 \}), \,
		(2, \{ 12, 14, 16, 17, 18, 19, 21, 23 \}), \,
		(3, \{ 13, 15, 17, 18, 19, 20, 22, 24 \}),$\\
		
		&& $(4, \{ 14, 16, 18, 19, 20, 21, 23, 25 \}), \,
		(5, \{15, 17, 19, 20, 21, 22, 24, 26 \}), \,
		(6, \{16, 18, 20, 21, 22, 23, 25, 27 \}),$\\
		
		&& $(7, \{17, 19, 21, 22, 23, 24, 26, 28 \}), \,
		(8, \{18, 20, 22, 23, 24, 25, 27, 29 \}), \,
		(9, \{19, 21, 23, 24, 25, 26, 28, 30 \}),$\\
		
		&& $(10, \{20, 22, 24, 25, 26, 27, 29, 31 \}), \,
		(11, \{21, 23, 25, 26, 27, 28, 30 \}), \,
		(12, \{22, 24, 26, 27, 28, 29, 31 \}),$\\
		
		&& $(13, \{23, 25, 27, 28, 29, 30 \}), \,
		(14, \{24, 26, 28, 29, 30, 31 \}), \,
		(15, \{25, 27, 29, 30, 31 \}), \,
		(16, \{26, 28, 30, 31 \}),$\\
		
		&& $(17, \{27, 29, 31 \}), \,
		(18, \{28, 30 \}), \,
		(19, \{29, 31 \}), \,
		(20, \{30 \}), \,
		(21, \{31 \})$\\
		
		%%% Layer 2
		$124$ & $(2,2)$ & $(32, \{42, 44, 46, 47, 48, 49, 51, 53 \}), \,
		(33, \{43, 45, 47, 48, 49, 50, 52, 54 \}), \,
		(34, \{44, 46, 48, 49, 50, 51, 53, 55 \}),$\\
		
		&& $(35, \{45, 47, 49, 50, 51, 52, 54, 56 \}), \,
		(36, \{46, 48, 50, 51, 52, 53, 55, 57 \}), \,
		(37, \{47, 49, 51, 52, 53, 54, 56, 58 \}),$\\
		
		&& $(38, \{48, 50, 52, 53, 54, 55, 57, 59 \}), \,
		(39, \{49, 51, 53, 54, 55, 56, 58, 60 \}), \,
		(40, \{50, 52, 54, 55, 56, 57, 59, 61 \}),$\\
		
		&& $(41, \{51, 53, 55, 56, 57, 58, 60, 62 \}), \,
		(42, \{52, 54, 56, 57, 58, 59, 61 \}), \,
		(43, \{53, 55, 57, 58, 59, 60, 62 \}),$\\
		
		&& $(44, \{54, 56, 58, 59, 60, 61 \}), \,
		(45, \{55, 57, 59, 60, 61, 62 \}), \,
		(46, \{56, 58, 60, 61, 62 \}), \,
		(47, \{57, 59, 61, 62 \}),$\\
		
		&& $(48, \{58, 60, 62 \}), \,
		(49, \{59, 61 \}),\,
		(50, \{60, 62 \}),\,
		(51, \{61 \}),\,
		(52, \{62 \})$\\
		
		%%% Layer 3
		$124$ & $(3,3)$ & 
		$(63, \{73, 75, 77, 78, 79, 80, 82, 84 \}), \, 
		(64, \{74, 76, 78, 79, 80, 81, 83, 85 \}), \, 
		(65, \{75, 77, 79, 80, 81, 82, 84, 86 \}),$\\ 
		
		&& $(66, \{76, 78, 80, 81, 82, 83, 85, 87 \}), \, 
		(67, \{77, 79, 81, 82, 83, 84, 86, 88 \}), \, 
		(68, \{78, 80, 82, 83, 84, 85, 87, 89 \}),$\\ 
		
		&& $(69, \{79, 81, 83, 84, 85, 86, 88, 90 \}), \, 
		(70, \{80, 82, 84, 85, 86, 87, 89, 91 \}), \, 
		(71, \{81, 83, 85, 86, 87, 88, 90, 92 \}),$\\
		
		&& $(72, \{82, 84, 86, 87, 88, 89, 91, 93 \}), \, 
		(73, \{83, 85, 87, 88, 89, 90, 92 \}), \, 
		(74, \{84, 86, 88, 89, 90, 91, 93 \}),$\\
		
		&& $(75, \{85, 87, 89, 90, 91, 92 \}), \, 
		(76, \{86, 88, 90, 91, 92, 93 \}), \, 
		(77, \{87, 89, 91, 92, 93 \}), \, 
		(78, \{88, 90, 92, 93 \}),$\\
		
		&& $(79, \{89, 91, 93 \}), \, 
		(80, \{90, 92 \}), \, 
		(81, \{91, 93 \}), \, 
		(82, \{92 \}), \, 
		(83, \{93 \})$\\
		
		%%% Layers 1 and 2
		$310$ & $(1,2)$ & 
		$(1, \{34, 39, 40, 42, 43, 45, 47, 48, 49, 51 \}), \, 
		(2, \{35, 40, 41, 43, 44, 46, 48, 49, 50, 52 \}), \, 
		(3, \{36, 41, 42, 44, 45, 47, 49, 50, 51, 53 \}),$\\
		
		&& $(4, \{37, 42, 43, 45, 46, 48, 50, 51, 52, 54 \}), \, 
		(5, \{38, 43, 44, 46, 47, 49, 51, 52, 53, 55 \}), \, 
		(6, \{39, 44, 45, 47, 48, 50, 52, 53, 54, 56 \}),$\\ 
		
		&& $(7, \{40, 45, 46, 48, 49, 51, 53, 54, 55, 57 \}), \, 
		(8, \{41, 46, 47, 49, 50, 52, 54, 55, 56, 58 \}), \, 
		(9, \{42, 47, 48, 50, 51, 53, 55, 56, 57, 59 \}),$\\ 
		
		&& $(10, \{43, 48, 49, 51, 52, 54, 56, 57, 58, 60 \}), \, 
		(11, \{44, 49, 50, 52, 53, 55, 57, 58, 59, 61 \}), \, 
		(12, \{45, 50, 51, 53, 54, 56, 58, 59, 60, 62 \}),$\\ 
		
		&& $(13, \{32, 46, 51, 52, 54, 55, 57, 59, 60, 61 \}), \, 
		(14, \{33, 47, 52, 53, 55, 56, 58, 60, 61, 62 \}), \, 
		(15, \{32, 34, 48, 53, 54, 56, 57, 59, 61, 62 \}),$\\ 
		
		&& $(16, \{32, 33, 35, 49, 54, 55, 57, 58, 60, 62 \}), \, 
		(17, \{32, 33, 34, 36, 50, 55, 56, 58, 59, 61 \}), \, 
		(18, \{33, 34, 35, 37, 51, 56, 57, 59, 60, 62 \}),$\\ 
		
		&& $(19, \{32, 34, 35, 36, 38, 52, 57, 58, 60, 61 \}), \, 
		(20, \{33, 35, 36, 37, 39, 53, 58, 59, 61, 62 \}), \, 
		(21, \{32, 34, 36, 37, 38, 40, 54, 59, 60, 62 \}),$\\ 
		
		&& $(22, \{32, 33, 35, 37, 38, 39, 41, 55, 60, 61 \}), \, 
		(23, \{33, 34, 36, 38, 39, 40, 42, 56, 61, 62 \}), \, 
		(24, \{32, 34, 35, 37, 39, 40, 41, 43, 57, 62 \}),$\\
		
		&& $(25, \{32, 33, 35, 36, 38, 40, 41, 42, 44, 58 \}), \, 
		(26, \{33, 34, 36, 37, 39, 41, 42, 43, 45, 59 \}), \, 
		(27, \{34, 35, 37, 38, 40, 42, 43, 44, 46, 60 \}),$\\ 
		
		&& $(28, \{35, 36, 38, 39, 41, 43, 44, 45, 47, 61 \}), \, 
		(29, \{36, 37, 39, 40, 42, 44, 45, 46, 48, 62 \}), \, 
		(30, \{32, 37, 38, 40, 41, 43, 45, 46, 47, 49 \}),$\\ 
		
		&& $(31, \{33, 38, 39, 41, 42, 44, 46, 47, 48, 50 \})$\\
		
		%%% Layers 1 and 3
		$310$ & $(1,3)$ & 
		$(1, \{75, 77, 78, 79, 81, 83, 84, 86, 87, 92 \}), \, 
		(2, \{76, 78, 79, 80, 82, 84, 85, 87, 88, 93 \}), \, 
		(3, \{63, 77, 79, 80, 81, 83, 85, 86, 88, 89 \}),$\\ 
		
		&& $(4, \{64, 78, 80, 81, 82, 84, 86, 87, 89, 90 \}), \, 
		(5, \{65, 79, 81, 82, 83, 85, 87, 88, 90, 91 \}), \, 
		(6, \{66, 80, 82, 83, 84, 86, 88, 89, 91, 92 \}), $\\
		
		&& $(7, \{67, 81, 83, 84, 85, 87, 89, 90, 92, 93 \}), \, 
		(8, \{63, 68, 82, 84, 85, 86, 88, 90, 91, 93 \}), \, 
		(9, \{63, 64, 69, 83, 85, 86, 87, 89, 91, 92 \}), $\\
		
		&& $(10, \{64, 65, 70, 84, 86, 87, 88, 90, 92, 93 \}), \, 
		(11, \{63, 65, 66, 71, 85, 87, 88, 89, 91, 93 \}), \, 
		(12, \{63, 64, 66, 67, 72, 86, 88, 89, 90, 92 \}), $\\
		
		&& $(13, \{64, 65, 67, 68, 73, 87, 89, 90, 91, 93 \}), \, 
		(14, \{63, 65, 66, 68, 69, 74, 88, 90, 91, 92 \}), \, 
		(15, \{64, 66, 67, 69, 70, 75, 89, 91, 92, 93 \}),$\\ 
		
		&& $(16, \{63, 65, 67, 68, 70, 71, 76, 90, 92, 93 \}), \, 
		(17, \{63, 64, 66, 68, 69, 71, 72, 77, 91, 93 \}), \, 
		(18, \{63, 64, 65, 67, 69, 70, 72, 73, 78, 92 \}),$\\ 
		
		&& $(19, \{64, 65, 66, 68, 70, 71, 73, 74, 79, 93 \}), \, 
		(20, \{63, 65, 66, 67, 69, 71, 72, 74, 75, 80 \}), \, 
		(21, \{64, 66, 67, 68, 70, 72, 73, 75, 76, 81 \}),$\\ 
		
		&& $(22, \{65, 67, 68, 69, 71, 73, 74, 76, 77, 82 \}), \, 
		(23, \{66, 68, 69, 70, 72, 74, 75, 77, 78, 83 \}), \, 
		(24, \{67, 69, 70, 71, 73, 75, 76, 78, 79, 84 \}),$\\ 
		
		&& $(25, \{68, 70, 71, 72, 74, 76, 77, 79, 80, 85 \}), \, 
		(26, \{69, 71, 72, 73, 75, 77, 78, 80, 81, 86 \}), \, 
		(27, \{70, 72, 73, 74, 76, 78, 79, 81, 82, 87 \}),$\\ 
		
		&& $(28, \{71, 73, 74, 75, 77, 79, 80, 82, 83, 88 \}), \, 
		(29, \{72, 74, 75, 76, 78, 80, 81, 83, 84, 89 \}), \, 
		(30, \{73, 75, 76, 77, 79, 81, 82, 84, 85, 90 \}),$\\ 
		
		&& $(31, \{74, 76, 77, 78, 80, 82, 83, 85, 86, 91 \})$\\ 
		
		%%% Layers 2 and 3
		$310$ & $(2,3)$ & 
		$(32, \{65, 70, 71, 73, 74, 76, 78, 79, 80, 82 \}), \, 
		(33, \{66, 71, 72, 74, 75, 77, 79, 80, 81, 83 \}), \, 
		(34, \{67, 72, 73, 75, 76, 78, 80, 81, 82, 84 \}), $\\
		
		&& $(35, \{68, 73, 74, 76, 77, 79, 81, 82, 83, 85 \}), \, 
		(36, \{69, 74, 75, 77, 78, 80, 82, 83, 84, 86 \}), \, 
		(37, \{70, 75, 76, 78, 79, 81, 83, 84, 85, 87 \}), $\\
		
		&& $(38, \{71, 76, 77, 79, 80, 82, 84, 85, 86, 88 \}), \, 
		(39, \{72, 77, 78, 80, 81, 83, 85, 86, 87, 89 \}), \, 
		(40, \{73, 78, 79, 81, 82, 84, 86, 87, 88, 90 \}), $\\
		
		&& $(41, \{74, 79, 80, 82, 83, 85, 87, 88, 89, 91 \}), \, 
		(42, \{75, 80, 81, 83, 84, 86, 88, 89, 90, 92 \}), \, 
		(43, \{76, 81, 82, 84, 85, 87, 89, 90, 91, 93 \}), $\\
		
		&& $(44, \{63, 77, 82, 83, 85, 86, 88, 90, 91, 92 \}), \, 
		(45, \{64, 78, 83, 84, 86, 87, 89, 91, 92, 93 \}), \, 
		(46, \{63, 65, 79, 84, 85, 87, 88, 90, 92, 93 \}), $\\
		
		&& $(47, \{63, 64, 66, 80, 85, 86, 88, 89, 91, 93 \}), \, 
		(48, \{63, 64, 65, 67, 81, 86, 87, 89, 90, 92 \}), \, 
		(49, \{64, 65, 66, 68, 82, 87, 88, 90, 91, 93 \}), $\\ 
		
		&& $(50, \{63, 65, 66, 67, 69, 83, 88, 89, 91, 92 \}), \, 
		(51, \{64, 66, 67, 68, 70, 84, 89, 90, 92, 93 \}), \, 
		(52, \{63, 65, 67, 68, 69, 71, 85, 90, 91, 93 \}),$\\ 
		
		&& $(53, \{63, 64, 66, 68, 69, 70, 72, 86, 91, 92 \}), \, 
		(54, \{64, 65, 67, 69, 70, 71, 73, 87, 92, 93 \}), \, 
		(55, \{63, 65, 66, 68, 70, 71, 72, 74, 88, 93 \}),$\\ 
		
		&& $(56, \{63, 64, 66, 67, 69, 71, 72, 73, 75, 89 \}), \, 
		(57, \{64, 65, 67, 68, 70, 72, 73, 74, 76, 90 \}), \, 
		(58, \{65, 66, 68, 69, 71, 73, 74, 75, 77, 91 \}),$\\ 
		
		&& $(59, \{66, 67, 69, 70, 72, 74, 75, 76, 78, 92 \}), \, 
		(60, \{67, 68, 70, 71, 73, 75, 76, 77, 79, 93 \}), \, 
		(61, \{63, 68, 69, 71, 72, 74, 76, 77, 78, 80 \}),$\\
		
		&& $(62, \{64, 69, 70, 72, 73, 75, 77, 78, 79, 81 \})$\\  
		\bottomrule 
	\end{tabular}
	}
	\label{list:93} 
\end{table}
%%%%% End of G 93

%%%% start of G_96
\begin{table}[!htbp]
	\caption{The $1\,680$ edges of $G_{96}$ in Proposition\ref{prop96}. Vertices are partitioned into $6$ layers $\mathcal{L}_1, \ldots, \mathcal{L}_{10}$, with $16$ vertices in each layer. Vertices $z,z+1,\ldots,z+15$ are in $\mathcal{L}_{k+1}$ for $0 \leq k \leq 5$ and $z=16k+1$.}
	\centering
	\resizebox{\textwidth}{!}{
	\begin{tabular}{cc l}
		\toprule
		$\#$ & $(r,s)$ & $(i,\{j \in J\})$ \\
		\midrule
		$42$ & $(1,1)$ & 
		$(1, \{4, 8, 12, 13, 16 \}), 
		(2, \{5, 8, 9, 11, 14, 16 \}), 
		(3, \{6, 8, 10, 13, 15 \}), 
		(4, \{7, 10, 11, 16 \}), 
		(5, \{10, 12, 14, 15 \}), 
		(6, \{7, 9, 12, 14 \}),$\\
		& & $((7, \{10, 14 \}),(8, \{11, 14, 15 \}),(9, \{12, 14, 16 \}), 
		(10, \{13, 16 \}),(11, \{16 \}),(12, \{13, 15 \}),(13, \{16 \})$\\
		
		$44$ & $(2,2)$ &
		$(17, \{ 22, 24, 26, 27, 29, 31 \}),  
		(18, \{ 19, 21, 24, 26, 30, 31, 32 \}), 
		(19, \{ 22, 26, 30, 31\}),  
		(20, \{ 23, 26, 27, 29, 32\}),  
		(21, \{ 24, 26, 28, 31\}),$ \\
		
		& & 
		$(22, \{ 25, 28, 29\}),(23, \{ 28, 30, 32 \}),
		(24, \{ 25, 27, 30, 32 \}),(25, \{ 28, 32 \}), 
		(26, \{ 29, 32 \}), (27, \{ 30, 32\}), (28, \{ 31 \}), (30, \{ 31 \})$\\
		
		$42$ & $(3,3)$ &
		$(33, \{36, 38, 40, 43, 45, 48 \}), 
		(34, \{37, 40, 41, 46, 48 \}), 
		(35, \{40, 42, 44, 45, 47 \}), 
		(36, \{37, 39, 42, 44, 48 \}), 
		(37, \{40, 44, 48 \})$,\\
		
		& &
		$(38, \{41, 44, 45, 47 \}),
		(39, \{42, 44, 46 \}), 
		(40, \{43, 46, 47 \}), 
		(41, \{46, 48 \}), 
		(42, \{43, 45, 48 \}), 
		(43, \{46 \}), 
		(44, \{47 \}), 
		(45, \{48 \})$\\
		
		$42$ & $(4,4)$ &
		$(49, \{52, 56, 60, 61, 64 \}), 
		(50, \{53, 56, 57, 59, 62, 64 \}), 
		(51, \{54, 56, 58, 61, 63 \}), 
		(52, \{55, 58, 59, 64 \}), 
		(53, \{58, 60, 62, 63 \})$,\\
		& &
		$(54, \{55, 57, 60, 62 \}), 
		(55, \{58, 62 \}), 
		(56, \{59, 62, 63 \}), 
		(57, \{60, 62, 64 \}), 
		(58, \{61, 64 \}), 
		(59, \{64 \}), 
		(60, \{61, 63 \}), 
		(61, \{64 \})$\\
		
		$44$ & $(5,5)$ &
		$(65, \{70, 72, 74, 75, 77, 79 \}), 
		(66, \{67, 69, 72, 74, 78, 79, 80 \}), 
		(67, \{70, 74, 78, 79 \}), 
		(68, \{71, 74, 75, 77, 80 \}), 
		(69, \{72, 74, 76, 79 \})$,\\
		& &
		$(70, \{73, 76, 77 \}), 
		(71, \{76, 78, 80 \}), 
		(72, \{73, 75, 78, 80 \}), 
		(73, \{76, 80 \}), 
		(74, \{77, 80 \}), 
		(75, \{78, 80 \}), 
		(76, \{79 \}), 
		(78, \{79 \})$\\
		
		$42$ & $(6,6)$ &
		$(81, \{84, 86, 88, 91, 93, 96 \}), 
		(82, \{85, 88, 89, 94, 96 \}), 
		(83, \{88, 90, 92, 93, 95 \}), 
		(84, \{85, 87, 90, 92, 96 \}), 
		(85, \{88, 92, 96 \})$,\\
		& & 
		$(86, \{89, 92, 93, 95 \}), 
		(87, \{90, 92, 94 \}), 
		(88, \{91, 94, 95 \}), 
		(89, \{94, 96 \}), 
		(90, \{91, 93, 96 \}), 
		(91, \{94 \}), 
		(92, \{95 \}), 
		(93, \{96 \})$\\
		
		$100$ & $(1,2)$ &
		$(1, \{17, 20, 21, 25, 26, 30, 32 \}), 
		(2, \{21, 22, 24, 26 \}), 
		(3, \{18, 22, 23, 26, 27, 28, 31 \}), 
		(4, \{18, 20, 23, 24, 28 \})$,\\
		&& 
		$(5, \{17, 19, 24, 28, 29, 30, 31 \}), 
		(6, \{18, 19, 20, 22, 30, 31, 32 \}), 
		(7, \{18, 19, 22, 23, 26, 27, 31, 32 \}), 
		(8, \{17, 20, 22, 27, 28, 30, 32 \})$,\\
		&& 
		$(9, \{19, 21, 24, 28, 29, 32 \}), 
		(10, \{17, 22, 24, 26, 29, 30 \}), 
		(11, \{18, 20, 21, 23, 25, 30 \}), 
		(12, \{18, 20, 24, 25, 26, 28 \})$,\\
		&& 
		$(13, \{20, 24, 25, 28, 29, 32 \}), 
		(14, \{17, 20, 21, 23, 26, 28 \}), 
		(15, \{18, 20, 22, 25, 27, 30 \}), 
		(16, \{19, 22, 23, 28, 30, 32 \})$\\
		
		$104$ & $(1,3)$ &
		$(1, \{33, 35, 37, 43, 44, 46 \}), 
		(2, \{33, 34, 37, 38, 39, 41, 43, 45 \}), 
		(3, \{34, 35, 39, 45, 46, 48 \}), 
		(4, \{35, 36, 37, 39, 40, 41, 45, 47 \})$,\\ 
		&&
		$(5, \{33, 36, 41, 44, 47, 48 \}), 
		(6, \{37, 39, 41, 42, 43, 47 \}), 
		(7, \{36, 38, 39, 41, 43 \}), 
		(8, \{39, 40, 43, 44, 45, 47 \})$,\\ 
		&&
		$(9, \{33, 34, 37, 40, 41, 45 \}), 
		(10, \{34, 41, 42, 43, 45, 46, 47 \}), 
		(11, \{34, 35, 36, 37, 39, 42, 47 \}), 
		(12, \{36, 37, 38, 43, 45, 47, 48 \})$,\\ 
		&&
		$(13, \{33, 37, 38, 42, 44, 45, 47 \}), 
		(14, \{33, 34, 36, 38, 45, 46 \}), 
		(15, \{34, 35, 38, 39, 40, 43, 46, 47 \}), 
		(16, \{35, 36, 40, 47, 48 \})$ \\
		%%%%%
		$70$ & $(1,4)$ &
		$(1, \{55, 56, 60, 61, 62, 64 \}), 
		(2, \{55, 59, 62 \}), 
		(3, \{56, 57, 58, 62, 63, 64 \}), 
		(4, \{55, 57, 64 \}), 
		(5, \{58, 59, 60, 62, 64 \}), 
		(6, \{57, 59 \})$,\\
		&&
		$(7, \{49, 50, 52, 61, 62 \}), 
		(8, \{49, 51, 61 \}),
		(9, \{51, 52, 54, 62, 63, 64 \}), 
		(10, \{51, 53, 61, 63 \}), 
		(11, \{50, 53, 54, 64 \})$,\\
		&&
		$(12, \{49, 53, 63 \}),(13, \{49, 55, 56, 58 \}), 
		(14, \{49, 50, 51, 53, 55, 57 \}), 
		(15, \{51, 57, 58, 60 \}), 
		(16, \{49, 51, 52, 53, 57, 59 \})$\\
		
		$100$ & $(1,5)$ &
		$(1, \{66, 69, 71, 72, 73 \}), 
		(2, \{67, 72, 73, 74, 75, 79 \}), 
		(3, \{66, 67, 68, 71, 75, 79, 80 \}), 
		(4, \{68, 69, 75, 76, 77, 80 \}),$\\ 
		&&
		$(5, \{65, 66, 67, 69, 70, 77 \}), 
		(6, \{66, 70, 71, 73, 77, 78, 80 \}), 
		(7, \{66, 67, 68, 70, 72, 75, 77, 78, 79 \}), 
		(8, \{65, 68, 73, 78, 79, 80 \}),$\\ 
		&&
		$(9, \{68, 69, 70, 72, 73, 74, 77 \}), 
		(10, \{70, 74, 75 \}), 
		(11, \{65, 66, 68, 70, 71, 72, 73, 75, 76 \}), 
		(12, \{65, 72, 76, 77, 79 \}),$\\ 
		&&
		$(13, \{67, 68, 72, 73, 74, 76, 78 \}), 
		(14, \{67, 71, 74, 79 \}), 
		(15, \{68, 69, 70, 74, 75, 76, 78, 79, 80 \}), 
		(16, \{67, 69, 76, 80 \})$\\
		
		$102$ & $(1,6)$ &
		$(1, \{83, 84, 85, 87, 94, 96 \}), 
		(2, \{82, 84, 86, 89, 90, 91, 92, 93 \}), 
		(3, \{87, 89, 92, 96 \}), 
		(4, \{81, 84, 85, 86, 88, 93, 94, 95 \}),$\\ 
		&&
		$(5, \{81, 82, 85, 89, 92, 94 \}), 
		(6, \{82, 83, 87, 88, 90, 91, 95, 96 \}), 
		(7, \{89, 90, 91, 93 \}), 
		(8, \{81, 85, 88, 90, 92, 95, 96 \}),$\\ 
		&&
		$(9, \{81, 85, 86, 93, 95 \}), 
		(10, \{81, 82, 83, 86, 87, 90, 91, 92, 94 \}), 
		(11, \{83, 87, 88, 91, 95 \}), 
		(12, \{83, 84, 86, 88, 89, 93, 94, 96 \}),$\\ 
		&&
		$(13, \{81, 83, 84, 85, 95, 96 \}), 
		(14, \{84, 85, 86, 87, 91, 94, 96 \}), 
		(15, \{83, 87, 91, 92 \}), 
		(16, \{81, 87, 88, 89, 92, 93, 96 \})$\\ 
		
		$100$ & $(2,3)$ &
		$(17, \{36, 40, 41, 42, 43, 45, 48 \}), 
		(18, \{34, 42, 43, 44 \}), 
		(19, \{34, 35, 38, 39, 43, 44, 48 \}), 
		(20, \{34, 39, 40, 42, 44 \}),$\\
		&&
		$(21, \{33, 36, 40, 41, 44, 45, 46 \}), 
		(22, \{34, 36, 38, 41, 42, 46 \}), 
		(23, \{33, 35, 37, 42, 46, 47, 48 \}), 
		(24, \{36, 37, 38, 40, 48 \}),$\\
		&&
		$(25, \{36, 37, 40, 41, 44, 45 \}), 
		(26, \{33, 35, 38, 40, 45, 46, 48 \}), 
		(27, \{34, 37, 39, 42, 46, 47 \}), 
		(28, \{34, 35, 40, 42, 44, 47, 48 \}),$\\ 
		&&
		$(29, \{34, 36, 38, 39, 41, 43, 48 \}), 
		(30, \{33, 36, 38, 42, 43, 44, 46 \}), 
		(31, \{34, 38, 42, 43, 46, 47 \}), 
		(32, \{35, 38, 39, 41, 44, 46 \})$\\ 
		
		$100$ & $(2,4)$ &
		$(17, \{53, 56, 59, 60 \}), 
		(18, \{49, 51, 53, 54, 55, 59 \}), 
		(19, \{50, 51, 53, 55, 61, 62, 64 \}), 
		(20, \{51, 52, 55, 56, 57, 59, 61, 63 \}),$\\
		&&
		$(21, \{49, 52, 53, 57, 63, 64 \}), 
		(22, \{53, 54, 55, 57, 58, 59, 63 \}), 
		(23, \{49, 51, 54, 59, 62 \}), 
		(24, \{49, 50, 55, 57, 59, 60, 61 \}),$\\
		&&
		$(25, \{49, 50, 54, 56, 57, 59, 61 \}), 
		(26, \{50, 57, 58, 61, 62, 63 \}), 
		(27, \{50, 51, 52, 55, 58, 59, 63 \}), 
		(28, \{52, 59, 60, 61, 63, 64 \}),$\\ 
		&&
		$(29, \{52, 53, 54, 55, 57, 60 \}), 
		(30, \{54, 55, 56, 61, 63 \}), 
		(31, \{50, 51, 55, 56, 60, 62, 63 \}), 
		(32, \{51, 52, 54, 56, 63, 64 \})$\\
		
		$72$ & $(2,5)$ &
		$(17, \{70, 71, 72, 74, 76, 77, 78, 79 \}), 
		(18, \{69, 71, 78 \}), 
		(19, \{73, 74, 78, 79, 80 \}), 
		(20, \{73, 77, 80 \}), 
		(21, \{66, 74, 75, 76, 80 \}),$\\ 
		&&
		$(22, \{65, 73, 75 \}),(23, \{65, 66, 76, 77, 78, 80 \}), 
		(24, \{65, 75, 77 \}), 
		(25, \{67, 68, 70, 79, 80 \}), 
		(26, \{65, 67, 69, 79 \}),\,
		(27, \{69, 70, 72, 80 \}),$\\ 
		&&
		$(28, \{65, 69, 71, 79 \}), 
		(29, \{65, 68, 71, 72 \}), 
		(30, \{65, 66, 67, 71 \}), 
		(31, \{65, 67, 73, 74, 76 \}), 
		(32, \{67, 68, 69, 71, 73, 75 \})$\\
		
		$100$ & $(2,6)$ &
		$(17, \{81, 82, 89, 93, 94 \}), 
		(18, \{82, 83, 85, 89, 90, 92, 94, 95 \}), 
		(19, \{82, 84, 87, 89, 90, 91 \}), 
		(20, \{85, 90, 91, 92, 93 \}),$\\ 
		&&
		$(21, \{81, 82, 84, 85, 86, 89, 93 \}), 
		(22, \{82, 86, 87, 93, 94, 95 \}), 
		(23, \{82, 83, 84, 85, 87, 88, 95 \}), 
		(24, \{84, 88, 89, 91, 95, 96 \}),$\\ 
		&&
		$(25, \{84, 85, 86, 88, 90, 93, 95, 96 \}), 
		(26, \{83, 86, 91, 96 \}), 
		(27, \{81, 82, 86, 87, 88, 90, 91, 92, 95 \}), 
		(28, \{81, 88, 92, 93 \}),$\\ 
		&&
		$(29, \{82, 83, 84, 86, 88, 89, 90, 91, 93, 94 \}), 
		(30, \{81, 83, 90, 94, 95 \}), 
		(31, \{85, 86, 90, 91, 92, 94, 96 \}), 
		(32, \{85, 89, 92 \})$\\
		
		$102$ & $(3,4)$ &
		$(33, \{52, 53, 56, 57, 58, 61, 64 \}), 
		(34, \{50, 53, 54, 58 \}), 
		(35, \{49, 54, 58, 59, 60, 61, 63 \}), 
		(36, \{49, 50, 52, 60, 61, 62 \}),$\\ 
		&&
		$(37, \{49, 52, 53, 56, 57, 61, 62 \}), 
		(38, \{50, 52, 57, 58, 60, 62 \}), 
		(39, \{49, 51, 54, 58, 59, 62, 63, 64 \}), 
		(40, \{52, 54, 56, 59, 60, 64 \}),$\\ 
		&&
		$(41, \{50, 51, 53, 55, 60, 64 \}), 
		(42, \{50, 54, 55, 56, 58 \}), 
		(43, \{50, 54, 55, 58, 59, 62, 63 \}), 
		(44, \{50, 51, 53, 56, 58, 63, 64 \}),$\\ 
		&&
		$(45, \{50, 52, 55, 57, 60, 64 \}), 
		(46, \{49, 52, 53, 58, 60, 62 \}), 
		(47, \{52, 54, 56, 57, 59, 61 \}), 
		(48, \{49, 51, 54, 56, 60, 61, 62, 64 \})$\\
		
		$100$ & $(3,5)$ &
		$(33, \{65, 69, 75, 76, 78 \}), 
		(34, \{65, 66, 67, 69, 70, 71, 75, 77 \}), 
		(35, \{66, 71, 74, 77, 78 \}), 
		(36, \{67, 69, 71, 72, 73, 77 \}),$\\ 
		&&
		$(37, \{66, 68, 69, 71, 73, 79, 80 \}), 
		(38, \{69, 70, 73, 74, 75, 77, 79 \}), 
		(39, \{67, 70, 71, 75 \}), 
		(40, \{71, 72, 73, 75, 76, 77 \}),$\\
		&&
		$(41, \{65, 66, 67, 69, 72, 77, 80 \}), 
		(42, \{66, 67, 68, 73, 75, 77, 78, 79 \}), 
		(43, \{67, 68, 72, 74, 75, 77, 79 \}), 
		(44, \{66, 68, 75, 76, 79, 80 \}),$\\ 
		&&
		$(45, \{65, 68, 69, 70, 73, 76, 77 \}), 
		(46, \{65, 66, 70, 77, 78, 79 \}), 
		(47, \{66, 70, 71, 72, 73, 75, 78 \}), 
		(48, \{72, 73, 74, 79 \})$\\
		
		$70$ & $(3,6)$ &
		$(33, \{86, 87, 88, 92, 93, 94, 96 \}), 
		(34, \{85, 87, 94 \}), 
		(35, \{88, 89, 90, 92, 94, 95, 96 \}), 
		(36, \{87, 89, 96 \}), 
		(37, \{82, 91, 92, 96 \}),$\\ 
		&&
		$(38, \{81, 91, 95 \}), 
		(39, \{81, 82, 84, 92, 93, 94 \}), 
		(40, \{81, 83, 91, 93 \}), 
		(41, \{83, 84, 94, 95, 96 \}), 
		(42, \{83, 93, 95 \}), 
		(43, \{85, 86, 88 \}),$\\ 
		&&
		$(44, \{81, 83, 85, 87 \}), 
		(45, \{81, 87, 88, 90 \}), 
		(46, \{81, 82, 83, 87, 89 \}), 
		(47, \{83, 86, 89, 90 \}), 
		(48, \{81, 83, 84, 85, 89 \})$\\
		
		$100$ & $(4,5)$ &
		$(49, \{65, 68, 69, 73, 74, 78, 80 \}), 
		(50, \{69, 70, 72, 74 \}), 
		(51, \{66, 70, 71, 74, 75, 76, 79 \}), 
		(52, \{66, 68, 71, 72, 76 \}),$\\ 
		&&
		$(53, \{65, 67, 72, 76, 77, 78, 79 \}), 
		(54, \{66, 67, 68, 70, 78, 79, 80 \}), 
		(55, \{66, 67, 70, 71, 74, 75, 79, 80 \}), 
		(56, \{65, 68, 70, 75, 76, 78, 80 \}),$\\ 
		&&
		$(57, \{67, 69, 72, 76, 77, 80 \}), 
		(58, \{65, 70, 72, 74, 77, 78 \}), 
		(59, \{66, 68, 69, 71, 73, 78 \}), 
		(60, \{66, 68, 72, 73, 74, 76 \}),$\\ 
		&&
		$(61, \{68, 72, 73, 76, 77, 80 \}), 
		(62, \{65, 68, 69, 71, 74, 76 \}), 
		(63, \{66, 68, 70, 73, 75, 78 \}), 
		(64, \{67, 70, 71, 76, 78, 80 \})$\\ 
		
		$104$ & $(4,6)$ &
		$(49, \{81, 83, 85, 91, 92, 94 \}), 
		(50, \{81, 82, 85, 86, 87, 89, 91, 93 \}), 
		(51, \{82, 83, 87, 93, 94, 96 \}), 
		(52, \{83, 84, 85, 87, 88, 89, 93, 95 \}),$\\
		&& 
		$(53, \{81, 84, 89, 92, 95, 96 \}), 
		(54, \{85, 87, 89, 90, 91, 95 \}), 
		(55, \{84, 86, 87, 89, 91 \}), 
		(56, \{87, 88, 91, 92, 93, 95 \}),$\\
		&&
		$(57, \{81, 82, 85, 88, 89, 93 \}), 
		(58, \{82, 89, 90, 91, 93, 94, 95 \}), 
		(59, \{82, 83, 84, 85, 87, 90, 95 \}), 
		(60, \{84, 85, 86, 91, 93, 95, 96 \}),$\\ 
		&&
		$(61, \{81, 85, 86, 90, 92, 93, 95 \}), 
		(62, \{81, 82, 84, 86, 93, 94 \}), 
		(63, \{82, 83, 86, 87, 88, 91, 94, 95 \}), 
		(64, \{83, 84, 88, 95, 96 \})$\\
		
		$100$ & $(5,6)$ &
		$(65, \{84, 88, 89, 90, 91, 93, 96 \}), 
		(66, \{82, 90, 91, 92 \}), 
		(67, \{82, 83, 86, 87, 91, 92, 96 \}), 
		(68, \{82, 87, 88, 90, 92 \}),$\\ 
		&&
		$(69, \{81, 84, 88, 89, 92, 93, 94 \}), 
		(70, \{82, 84, 86, 89, 90, 94 \}), 
		(71, \{81, 83, 85, 90, 94, 95, 96 \}), 
		(72, \{84, 85, 86, 88, 96 \}),$\\ 
		&&
		$(73, \{84, 85, 88, 89, 92, 93 \}), 
		(74, \{81, 83, 86, 88, 93, 94, 96 \}), 
		(75, \{82, 85, 87, 90, 94, 95 \}), 
		(76, \{82, 83, 88, 90, 92, 95, 96 \}),$\\ 
		&&
		$(77, \{82, 84, 86, 87, 89, 91, 96 \}), 
		(78, \{81, 84, 86, 90, 91, 92, 94 \}), 
		(79, \{82, 86, 90, 91, 94, 95 \}), 
		(80, \{83, 86, 87, 89, 92, 94 \})$\\ 
		\bottomrule
	\end{tabular}
	}
	\label{table:List96}
\end{table}
%%%%%%%%%%%%%%%%%% End of Appendix

\medskip
% The data information below will be filled by AIMS editorial staff
Received xxxx 20xx; revised xxxx 20xx.
\medskip

\end{document}